\documentclass[1p,times]{elsarticle}

\usepackage[utf8]{inputenc}
\usepackage[T1]{fontenc}
\usepackage[english]{babel}

\usepackage{amsmath}
\usepackage{amssymb}
\usepackage{amsthm}
\usepackage{thmtools} 
\usepackage{mathtools}
\mathtoolsset{showonlyrefs, showmanualtags}
\usepackage{algorithm}
\usepackage{algorithmicx}
\usepackage[noend]{algpseudocode}
\usepackage{tikz}
\usepackage{tikz-cd}
\usepackage{multirow}
\usepackage{refcount}
\usepackage[textwidth=3.7cm, obeyDraft]{todonotes}

\usepackage[final]{hyperref}

\newcommand{\todonote}[2][]{\todo[color=green,#1]{#2}}

\newcommand{\tvnote}[2][]{\todonote[#1]{Th: #2}}

\newcommand{\chnote}[2][]{\todonote[#1]{Cl: #2}}

\usetikzlibrary{arrows, chains, positioning, decorations.pathreplacing, matrix, calc, fit}

\newtheorem{theorem}{Theorem}
\newtheorem{proposition}[theorem]{Proposition}
\newtheorem{corollary}[theorem]{Corollary}
\newtheorem{lemma}[theorem]{Lemma}
\newtheorem{conjecture}[theorem]{Conjecture}

\theoremstyle{definition}
\newtheorem{definition}[theorem]{Definition}
\newtheorem{example}[theorem]{Example}

\theoremstyle{remark}
\newtheorem{remark}[theorem]{Remark}

\renewcommand{\algorithmicrequire}{{\bf Input:}}
\renewcommand{\algorithmicensure}{{\bf Output:}}

\def\<#1>{\langle#1\rangle}
\def\Fr{\Sigma}
\def\MFr{M(\Fr)}

\let\set\mathbb
\newcommand\Q{\set Q}
\newcommand\N{\set N}

\newcommand\M{M}
\DeclareMathOperator{\lt}{lt}
\DeclareMathOperator{\lm}{lm}
\DeclareMathOperator{\lc}{lc}
\DeclareMathOperator{\coeff}{coeff}
\DeclareMathOperator{\supp}{supp}
\DeclareMathOperator{\sigc}{\mathfrak{sc}}
\DeclareMathOperator{\s}{\mathfrak{s}}
\DeclareMathOperator{\st}{\mathfrak{st}}

\DeclareMathOperator{\spol}{sp}
\DeclareMathOperator{\Syz}{Syz}

\newcommand{\Htriv}{H_{\textup{triv}}}

\newcommand\MMA{\mbox{\textsc{Mathematica}}}

\usepackage{xspace}
\newcommand{\textwhile}{``\emph{while}''\xspace}

\begin{document}

\begin{frontmatter}

\title{Signature Gr\"obner bases, bases of syzygies\\and cofactor reconstruction in the free algebra}

\author{Clemens Hofstadler}
\ead{clemens.hofstadler@jku.at}

\author{Thibaut Verron}
\ead{thibaut.verron@jku.at}

\address{Institute for Algebra, Johannes Kepler University Linz, Austria}


\begin{abstract}
Signature-based algorithms have become a standard approach for computing Gr\"obner bases in commutative polynomial rings.
However, so far, it was not clear how to extend this concept to the setting of noncommutative polynomials in the free algebra.
In this paper, we present a signature\nobreakdash-based algorithm for computing Gröbner bases in precisely this setting.
The algorithm is an adaptation of Buchberger's algorithm including signatures.
We prove that our algorithm correctly enumerates a signature Gr\"obner basis as well as a Gr\"obner basis of the module generated by the leading terms of the generators' syzygies, and that it terminates whenever the ideal admits a finite signature Gröbner basis.
Additionally, we adapt \mbox{well-known} signature-based criteria eliminating redundant reductions, such as the syzygy criterion, the F5 criterion and the singular criterion, to the case of noncommutative polynomials.
We also generalize reconstruction methods from the commutative setting that allow to recover, from partial information about signatures, the coordinates of elements of a Gröbner basis in terms of the input polynomials, as well as a basis of the syzygy module of the generators.
We have written a toy implementation of all the algorithms in the \MMA\ package \texttt{OperatorGB} and we compare our signature-based algorithm to the classical Buchberger algorithm for noncommutative polynomials.
\end{abstract}

\end{frontmatter}

\section{Introduction}

Gröbner bases have become a fundamental and multi-purpose tool in computational algebra.
They were initially introduced~\cite{Buchberger-1965} to answer questions about ideals of multivariate (commutative) polynomials, and they were subsequently generalized to several noncommutative settings, including Weyl polynomials~\cite{Galligo-1985-some-algorithmic-questions} encoding, for example, differential equations, and noncommutative polynomials in the free algebra~\cite{Mor85}, 
which can model matrix identities and, more generally, identities of linear operators.
In the latter setting, Gröbner bases and the associated reduction machinery are notably useful for simplifying and proving operator identities~\cite{HW94,HSW98,HRR19,CHRR20,SL20,RRH21}.

The main theoretical results needed to adapt the concept of Gr\"obner bases to the free algebra are due to Bergman~\cite{Ber78}, who used the abstract concept of reduction systems to generalize the ideas from the commutative setting.
This development was independent of the commutative theory.
Around the same time, also Bokut'~\cite{Bok76} proved statements that are essentially equivalent to the ones by Bergman.
Bokut' attributes his results to Shirshov, who had published similar results in the context of Lie algebras~\cite{Sir62}.
The first explicit algorithm for computing noncommutative Gr\"obner bases was proposed by Mora~\cite{Mor85} a few years later, who adapted Buchberger's algorithm to the free algebra.
Later, Mora also managed to unify the theory of Gr\"obner bases for commutative and noncommutative polynomial rings via a generalization of the Gaussian elimination algorithm~\cite{Mor94}.
Recently, also the F4 algorithm~\cite{Faugere-1999-F4} has been adapted to this setting~\cite{Xiu12}.
However, contrary to the commutative or the Weyl case, not all ideals in the free algebra admit a finite Gröbner basis.
Instead, the algorithms always enumerate a Gröbner basis, with termination if and only if a finite Gröbner basis exists w.r.t.\ the chosen monomial order.

In the case of commutative polynomials, the latest generation of Gröbner basis algorithms are the so-called signature-based algorithms, heralded by the F5 algorithm~\cite{Faugere-2002-F5}.
This class of algorithms was the subject of extensive research in the past 20 years, a survey of which can be found in~\cite{EF17}.
Those algorithms compute, in addition to a Gröbner basis, some information on how the polynomials in that basis were computed.
Using this information, the algorithms are able to identify relations between the computed polynomials, and use them to predict and avoid reductions to zero and redundant computations.
This yields a significant performance improvement.

Additionally, it was observed recently~\cite{Sun11,Gao-2015-new-framework-for} that the data of signatures is enough to reconstruct the cofactors of the Gröbner basis, that is, the coordinates of the elements of the basis in terms of the input polynomials.
Similarly, one can also compute a basis of the syzygy module of the generators.
Alternatively, those operations can be realised using the classical theory and algorithms for Gröbner bases of modules (one can see~\cite[Ch.~10]{Becker-1993} for a textbook exposition with references, or~ \cite[Ch.~3]{Adams-1994-introduction-to-groebner}), but signatures allow to reduce the cost of those computations to that of a Gröbner basis of an ideal.

Algorithms for computing Gröbner bases with signatures have also been developed in the case of Weyl algebras~\cite{Sun-2012-signature-groebner-solvable}, with the same application to the computation of coordinates of Gröbner basis elements, and of syzygies.
\chnote{Added reference}
Furthermore, in~\cite{Kin14} a noncommutative version of the F5 algorithm for right modules over quotients of path algebras was described.
In this setting, the information encoded by the signatures is used to efficiently compute bases of Loewy layers.
Finally, in the context of the free algebra, recent work~\cite{ChenavierLeonardVaccon2021} has independently introduced similar definitions as in Section~\ref{sec:sign-labell-grobn}, and proved lower bounds for the complexity of the set of leading monomials of the module of syzygies, in the sense of the Chomsky hierarchy.

The problem of computing the cofactors of a Gr\"obner basis is also central when working with noncommutative polynomials.
In particular, when proving operator identities, this information allows to construct a proof certificate for a given identity, which can be checked easily and independently of how it was obtained, see for example~\cite{Hof20}.
Like in the commutative case, the classical theory and algorithms for computing Gröbner bases in modules can also be used to obtain such information~\cite{BK06,Mor16}, but those algorithms are significantly more expensive than a mere Gröbner basis computation.
Furthermore, for most ideals in the free algebra, the module of syzygies is not finitely generated and does not admit a finite Gröbner basis, which makes those algorithms in fact only enumeration procedures.

In this paper, we show how to define and compute signature Gröbner bases for noncommutative polynomials in the free algebra, and we show how to use them to reconstruct the module representation of elements of the ideal, and a basis of the syzygy module of the generators.
We also generalize some classical signature-based criteria such as the syzygy criterion, the F5 criterion and the singular criterion, in order to use signatures to accelerate the algorithms.
\chnote{Link to algorithms}
More precisely, for introductory purposes, we first present Algorithm~\ref{algo LabelledGB}.
This algorithm is impractical because it has to perform expensive module computations.
Replacing those computations with signature manipulations naturally leads to Algorithm~\ref{algo SigGB}, which also includes the signature-based criteria.
This algorithm can be considered as a generic template for signature-based algorithms in the free algebra.
Finally, we introduce Algorithms~\ref{algo 2} and~\ref{algo 3} which allow to reconstruct the output of Algorithm~\ref{algo LabelledGB} using only the signatures.

A difficulty specific to the case of noncommutative polynomials is that some ideals may not have a finite signature Gröbner basis, even if the ideal has a finite Gröbner basis.
This is unavoidable, but we prove that the algorithms nonetheless correctly enumerate a signature Gröbner basis.

Additionally, as already mentioned, the module of syzygies of the generators usually does not admit a finite Gröbner basis. More precisely, the module spanned by the leading terms of the so-called \mbox{``trivial syzygies''} (sometimes called Koszul syzygies, or principal syzygies) is typically not finitely generated.
On the other hand, a strength of signature-based algorithms is precisely that they make it possible to identify those trivial syzygies, and in particular this set of trivial syzygies admits a finite and effective representation in terms of a signature Gröbner basis.
In classical cases, avoiding those trivial syzygies is the crux of the F5 criterion, and in the noncommutative case, it allows the algorithm to enumerate a basis of the syzygy module by only considering the non-trivial syzygies.

If our algorithm terminates, it computes a finite signature Gröbner basis, and in particular, a finite Gröbner basis and a finite and effective description of the module of syzygies of the input polynomials.
We conjecture (see Conjecture~\ref{sec:conj-termination}) that conversely, the existence of a finite Gröbner basis of the ideal and of a finite description of the module of syzygies of the generators, implies the existence of a finite signature Gröbner basis.

We also provide a toy implementation\footnote{Available at \url{https://clemenshofstadler.com/software/}} of the algorithms presented in this paper in the \MMA\ package \texttt{OperatorGB}~\cite{HRR19,Hof20}.
For an overview on other available software packages in the realm of noncommutative Gr\"obner bases, see~\cite{LSZ20} and references therein.

We show experimentally that the use of signatures allows to drastically reduce the number of S-polynomials considered and reduced to zero.
\chnote{Added something about efficiency}
However, first timings indicate that our implementation cannot compete with the standard noncommutative Buchberger algorithm.
So, while conceptually our algorithm is fairly simple and very similar to Buchberger's algorithm, implementing it in an \emph{efficient} way seems to be a highly non-trivial task.
However, we are optimistic that, as in the commutative case, noncommutative signature-based algorithms can lead to an acceleration of Gr\"obner basis computations in the free algebra.
In this work, we focus on building a theoretical foundation for this kind of algorithms and on highlighting possible uses of the additional information provided by the signatures.



\section{Preliminaries}

For the convenience of the reader, we recall the most important aspects of the theory of Gr\"obner bases in the free algebra and of Gr\"obner bases of submodules of the free bimodule in this section.
Additionally, we introduce the notion of \emph{signatures} in this noncommutative setting as a straightforward generalization of signatures from the commutative case.

We fix a finite set of indeterminates $X = \{x_1,\dots,x_n\}$ and denote by $\<X>$ the free monoid over $X$ containing all \emph{words} (or \emph{monomials}) of the form $w = x_{i_1}\dots x_{i_k}$ including the \emph{empty word} 1.
The quantity $k$ is called the \emph{length} of $w$.

For a field $K$, we let
\[
	K\<X> = \left\{ \sum_{w \in \<X>} c_w w \mid c_w \in K \text{ such that only finitely many }c_w \neq 0\right\}
\]
be the \emph{free algebra} generated by $X$ over $K$.
We consider the elements in $K\<X>$ as noncommutative polynomials with coefficients in $K$ and indeterminates in $X$, where indeterminates commute with coefficients but not with each other.

For a given set of polynomials $F \subseteq K\<X>$, we denote by $(F)$ the \emph{(two-sided) ideal} generated by $F$, that is
\[
	(F) = \left\{\sum_{i=1}^d a_i f_i b_i \mid f_i\in F,\; a_i, b_i\in K\<X>,\; d\in\N\right\}.
\]
The set $F$ is called a \emph{set of generators} of $(F)$. An ideal $I \subseteq K\<X>$ is said to be \emph{finitely generated} if there exists a finite set of generators $F \subseteq K\<X>$ such that $I = (F)$. 
We agree upon the convention to write $(f_1,\dots,f_r)$ instead of $(\{f_1,\dots,f_r\})$  if the elements of $F = \{f_1,\dots,f_r\}$ are given explicitly.

\begin{remark}
If $|X| > 1$, the free algebra $K\<X>$ is not Noetherian, i.e., there exist ideals in $K\<X>$ which are not finitely generated.
One prominent example is the ideal $(xy^ix \mid i\in\N) \subseteq K\langle x,y\rangle$, which has no finite set of generators.
\end{remark}

\begin{definition}
A \emph{monomial ordering} on $\<X>$ is a well-ordering $\preceq$ that is compatible with the multiplication in $\<X>$, that is,
$w \preceq w'$ implies $awb \preceq aw'b$ for all $a,b,w,w'\in\<X>$.
\end{definition}

An example of a monomial ordering on $\<X>$ is the degree lexicographic ordering $\preceq_{\textup{deglex}}$,
where two words $w,w' \in \<X>$ are first compared by their length and ties are broken by comparing the variables in $w$ and $w'$ from left to right using the lexicographic ordering $x_1 \prec_{\textup{lex}} \dots \prec_{\textup{lex}} x_n$.

In what follows, we fix a monomial ordering $\preceq$ on $\<X>$.
Then, every non-zero $f\in K\<X>$ has a unique representation of the form $f = c_1 w_1 + \dots + c_d w_d$ with $c_1,\dots,c_d \in K\setminus\{0\}$ and $w_1,\dots,w_d \in \<X>$ such that $w_1 \succ \dots \succ w_d$.

\begin{definition}
Let $f = c_1 w_1 + \dots + c_d w_d \in K\<X>\setminus\{0\}$ with $c_1,\dots,c_d \in K\setminus\{0\}$ and $w_1,\dots,w_d \in \<X>$ such that
$w_1 \succ \dots \succ w_d$.
Then, $w_1$ is called the \emph{leading monomial} of $f$, denoted by $\lm(f)$.
The coefficient $c_i$ of $w_i$ is denoted by $\coeff(f,w_i)$ for $i = 1,\dots,d$.
We call $c_1$ the \emph{leading coefficient} of $f$, abbreviated as $\lc(f)$.
If $\lc(f) = 1$, then $f$ is called \emph{monic}.
Furthermore, the \emph{leading term} $\lt(f)$ of $f$ is $\lt(f) = \lc(f)\cdot\lm(f)$.
Finally, the set $\{w_1,\dots,w_d\}$ is called the \emph{support} of $f$ and denoted by $\supp(f)$.
\end{definition}

We use the convention that the leading term and leading coefficient of the zero polynomial are $0$.
The leading monomial of the zero polynomial is undefined.

In the following, we briefly recall the most important results about Gr\"obner bases in $K\<X>$. 
For a more extensive treatment of this subject, we refer to the recent surveys~\cite{Xiu12,Mor16,Hof20}.
The main concept needed to discuss and compute noncommutative Gr\"obner bases is polynomial reduction.

\begin{definition}
Let $f,f',g \in K\<X>$ with $g \neq 0$. We say that $f$ \emph{reduces} to $f'$ by $g$ if there exist $a,b\in \<X>$ such that
$a\lm(g)b \in \supp(f)$ and
\[
	f' = f - \frac{\coeff(f, a\lm(g)b)}{\lc(g)}\cdot  agb.
\]
In this case, we write 
$f \rightarrow_g f'$.
\end{definition}

Based on this concept, for a set $G \subseteq K\<X>$, we define a reduction relation $\rightarrow_G\, \subseteq K\<X> \times K\<X>$ by 
$f \rightarrow_G f'$ if there exists $g \in G$ such that $f \rightarrow_g f'$.
We denote by $\overset{*}{\rightarrow}_G$ the reflexive, transitive closure of $\rightarrow_G$.
Using this reduction relation, we can now define Gr\"obner bases in $K\<X>$.

\begin{definition}
Let $I\subseteq K\<X>$ be an ideal and $G\subseteq I$.
Then, $G$ is a \emph{Gr\"obner basis} of $I$ if $f \overset{*}{\rightarrow}_G 0$ for all $f\in I$.
\chnote{extended this definition and added the paragraph afterwards.}
Furthermore, $G$ is called \emph{reduced} if all elements in $G$ are monic and no $g \in G$ is reducible by $G \setminus \{g\}$.
\end{definition}

We note that not all finitely generated ideals in $K\<X>$ have a finite Gr\"obner basis as witnessed by the following example.
For this example, we recall the well-known facts that the reduced Gr\"obner basis $G$ of an ideal $I \subseteq K\<X>$ is unique (w.r.t.\ a fixed monomial ordering) and that
$I$ has a finite Gr\"obner basis if and only if $G$ is finite (see for example~\cite[Prop.~3.3.17, Cor.~3.3.18]{Xiu12}).

\begin{example}\label{ex infinite gb}
\chnote{ok?}
Let $K$ be a field and $X = \{x,y\}$.
The principal ideal $I = (xyx - xy) \subseteq K\<X>$ does not have a finite Gr\"obner basis for any monomial ordering  $\preceq$ on $\<X>$.
In fact, the reduced Gr\"obner basis of $I$ is given by the infinite set $G = \{xy^nx - xy^n \mid n\geq 1\}$.
To see this, we first note that $g_n = xy^nx - xy^n \in I$ for all $n \geq 1$, which follows inductively from $g_1 = xyx -xy \in I$ and $g_{n+1} = xy g_n + g_1 (y^n- y^nx)$.
Consequently, we get that $G \subseteq I$.
Furthermore, since $xyx = (xy)\cdot x \succ (xy)\cdot 1 = xy$, all elements in $G$ are monic.
Also, no element in $G$ can be reduced by any other element, which shows that $G$ is reduced.
Then, one can verify that $G$ is indeed a Gr\"obner basis using the noncommutative analogue of Buchberger's S\nobreakdash-polynomial criterion, also known as Bergman's diamond lemma~\cite[Thm.~1.2]{Ber78}.
We note that each pair $g_i, g_j \in G$ leads to an S-polynomial
$s_{i,j} = xy^{i+j}x - xy^ixy^j$,
which can be reduced to zero as follows:
$ xy^{i+j}x - xy^ixy^j \,\rightarrow_{g_i}\, xy^{i+j}x - xy^{i+j} \,\rightarrow_{g_{i+j}}\,0.$
\end{example}


Even though ideal membership is undecidable in the free algebra, a noncommutative analog of Buchberger's algorithm can be used to enumerate a (possibly infinite) Gr\"{o}bner basis.

We also need the notion of Gr\"obner bases of  sub-bimodules of a free $K\<X>$-bimodule.
Hence, we shall briefly recall this concept in the following.
For further details on this topic, we refer to~\cite{Xiu12,Mor16}.
We note that when speaking about a (sub)module, we always mean a \mbox{(sub-)bimodule}.

For a fixed $r \in \N$, we denote by $\Fr = (K\<X> \otimes K\<X>)^r$ the \emph{free $K\<X>$-bimodule} of rank $r$ with the canonical basis $\varepsilon_1,\dots,\varepsilon_r$, where
$\varepsilon_i = (0,\dots,0,1 \otimes 1,0,\dots,0)$ with $1 \otimes 1$ appearing in the $i$-th position for $i=1,\dots,r$.
Furthermore, we let $\MFr = \{a\varepsilon_ib \mid a,b \in \<X>, 1 \leq i \leq r\}$ be the set of \emph{module monomials} in $\Fr$.
This set forms a $K$-vector space basis of $\Fr$.

\begin{definition}
A \emph{module ordering} on $\MFr$ is a well-ordering $\preceq_\M$ that is compatible with the scalar multiplication, that is,
$\mu \preceq_\M \mu'$ implies $a\mu b \preceq_\M a\mu'b$ for all $\mu,\mu'\in\MFr$ and $a,b\in\<X>$.
\end{definition}

Given a monomial ordering $\preceq$, an example of a module ordering is the term-over-position ordering $\preceq_{\textbf{top}}$, where $a_1 \varepsilon_{i_{1}} b_1 \preceq_{\textbf{top}} a_2 \varepsilon_{i_{2}} b_2$ for two module monomials $a_1 \varepsilon_{i_{1}} b_1, a_2 \varepsilon_{i_{2}} b_2 \in \MFr$ if one of the following conditions holds:
\begin{enumerate}
	\item $a_1 b_1 \prec a_2 b_2$;
	\item $a_{1} b_1 = a_2 b_2$ and $a_1 \prec a_2$;
	\item $a_1 b_1 = a_2 b_2$ and $a_1 = a_2$ and $i_{1} \leq i_{2}$.
\end{enumerate}

In the following, we fix a module ordering $\preceq_\M$ on $\MFr$. 
Then, using the $K$-vector space basis $\MFr$, every non-zero $\alpha \in \Fr$ can be uniquely written as $\alpha = c_1 \mu_1 + \dots + c_d \mu_d$ with $c_1,\dots,c_d \in K\setminus\{0\}$ and $\mu_1,\dots,\mu_d \in \MFr$ such that $\mu_1 \succ_\M \dots \succ_\M \mu_d$.

\begin{definition}
Let $\alpha = c_1 \mu_1 + \dots + c_d \mu_d$ with $c_1,\dots,c_d \in K\setminus\{0\}$ and $\mu_1,\dots,\mu_d \in \MFr$ such that $\mu_1 \succ_\M \dots \succ_\M \mu_d$.
Then, $\mu_1$ is called the \emph{signature} of $\alpha$, denoted by $\s(\alpha)$.
Analogously to the case of polynomials, we also define the \emph{signature coefficient} $\sigc(\alpha)$ of $\alpha$ to be the coefficient appearing in front of $\s(\alpha)$,
that is, $\sigc(\alpha) = c_1$.
Furthermore, the \emph{signature term} $\st(\alpha)$ of $\alpha$ is $\st(\alpha) = \sigc(\alpha)\cdot \s(\alpha)$.
\end{definition}

As for polynomials, we use the convention that the signature coefficient and the signature term of the zero element in $\Fr$ are $0$.
The signature of the zero element remains undefined.
By a slight abuse of notation, for a set $F \subseteq \Fr$, we denote by $\s(F)$ the set of signatures of all non-zero $\alpha \in F$, i.e., $\s(F) = \{ \s(\alpha) \mid 0 \neq \alpha \in F\}$.

\begin{remark}
\chnote{Added this remark and the lemma afterwards.}
Note that we define the signature $\s(\alpha)$ \emph{of a module element} $\alpha \in \Fr$.
This is different to the original definition of a signature.
Initially, signatures were looked at from a polynomial point of view and therefore signatures \emph{of polynomials} were defined~\cite{Faugere-2002-F5}.
The notion of signature which we use was first introduced in~\cite{G2V}.
\end{remark}

One immediate consequence of the definition of signatures is the following lemma.

\begin{lemma}
Let $\alpha \in \Fr$ be non-zero and let $a,b \in \<X>$. 
Then, $\s(a \alpha b)  = a \s(\alpha) b$.
\end{lemma}

\begin{definition}
Let $M \subseteq \Fr$ be a $K\<X>$-submodule.
A subset $G \subseteq M\setminus\{0\}$ is called a \emph{Gr\"obner basis} of $M$ if
\[
	\s(M) = \{a \s(\gamma) b \mid a,b \in \<X>, \gamma \in G\}.
\]
\end{definition}

\begin{remark}
  Gr\"obner bases of submodules can also be defined using a notion of reduction.
  The two definitions are equivalent, and the proof is the same as in the classical case of ideals.
  In our case, the present definition will be the more useful one for modules.
\end{remark}

\section{Signature and labelled Gr\"obner bases}
\label{sec:sign-labell-grobn}
The aim of this section is to introduce the notion of \emph{signature Gr\"obner bases} of ideals in the free algebra. 
As an intermediate notion, we define the concept of \emph{labelled Gr\"obner bases}, which are Gröbner bases keeping track of the construction of each element in terms of the generators.
As in the commutative case, they are defined using a more restrictive notion of polynomial reduction, called \emph{$\s$-reduction}.
Moreover, we also define and characterize noncommutative \emph{minimal} signature/labelled Gr\"obner bases.
We note that all notions introduced here are straightforward generalization of the same notions for commutative polynomials.
The key differences between the commutative and noncommutative case will become apparent in Section~\ref{sec:comp-sign-grobn}.

For the rest of this paper, we fix a finite indexed family of polynomials $(f_1,\dots,f_{r}) \in K\<X>^{r}$ generating an ideal $I = (f_1,\dots,f_r)$.
Furthermore, we fix a monomial ordering $\preceq$ on $\<X>$ and a module ordering $\preceq_\M$ on $\MFr$.
We additionally require the following two conditions:
\begin{enumerate}
	\item $\preceq$ and $\preceq_\M$ have to be \emph{compatible} in the sense that 
	\[
		a \prec b \iff a \varepsilon_i \prec_\M b \varepsilon_i \iff \varepsilon_i a \prec_\M \varepsilon_i b
	\]
	for all $a,b \in \<X>$ and $i = 1,\dots,r$.
	\item $\preceq_\M$ has to be \emph{fair}, meaning that the set $\{ \mu' \in \MFr \mid \mu' \prec_\M \mu\}$ has to be finite for all $\mu \in \MFr$.
\end{enumerate}

\begin{example}
  The module ordering $\preceq_{\mathbf{top}}$ with underlying monomial ordering $\preceq_{\textup{deglex}}$ is fair.
  Furthermore, in that case, the orderings $\preceq_{\textup{deglex}}$ and $\preceq_{\textbf{top}}$ are compatible.
  On the other hand, if the rank of the module is at least $2$, any position-over-term module ordering $\preceq_{\mathbf{pot}}$, that is, any ordering where $i_1 < i_2$ implies $a_1\varepsilon_{i_1}b_1 \prec_{\mathbf{pot}} a_2\varepsilon_{i_2}b_2$ for all $a_1,b_1,a_2,b_2 \in \langle X \rangle$, is not fair.
  Indeed, the set $\{\mu \in \MFr \mid \mu \prec_{\mathbf{pot}} \varepsilon_{2}\}$ is infinite.
\end{example}

\chnote{first version. needs to be polished}
We note that the requirement for a fair module ordering is a particularity of the noncommutative case.
To compute Gr\"obner bases (without signatures) in the free algebra using Buchberger's algorithm, a so-called \emph{fair} selection strategy has to be used.
Such a selection strategy ensures that every S-polynomial that is formed is eventually processed.
Using a non-fair selection strategy can cause the algorithm to run indefinitely, even if the ideal admits a finite Gr\"obner basis (w.r.t.\ the used monomial ordering).
Furthermore, in such cases, the infinite set produced by the algorithm fails to be a Gr\"obner basis, see~Example~47.6.27 and the subsequent discussion in~\cite{Mor16}.

Transferring the idea of a fair selection strategy to the case of signature-based algorithms leads to our definition of a fair module ordering.
Using such a fair ordering guarantees that the selection strategy used by our algorithm is fair. 
It is not clear whether or to what extent this requirement can be weakened.

From now on, we shall denote both orders $\preceq_\M$ and $\preceq$ by the same symbol $\preceq$.
It will be clear from the context which ordering is meant, as we will denote elements from $\Fr$ by Greek letters and elements from $K\<X>$ by Roman letters.
We note that all results that follow from here on depend (implicitly) on $\preceq$ and are to be understood w.r.t.\ our fixed monomial and module ordering.

Elements in the free $K\<X>$-bimodule $\Fr$ encode elements of the ideal $I$ via the $K\<X>$-module homomorphism
\begin{align*}
	\overline{\cdot} : \Fr \to K\<X>, \quad \alpha = \sum_i c_i a_i \varepsilon_{j_i} b_i \mapsto \overline \alpha \coloneqq \sum_i c_i a_i f_{j_i} b_i,
\end{align*}
with $a_i, b_i \in \<X>$ and $j_i \in \{1,\dots,r\}$.

We adapt the notation from~\cite{Sun11} and denote by $f^{[\alpha]}$ a pair $(f,\alpha) \in K\<X> \times \Fr$ with $f = \overline \alpha$.
We refer to $f^{[\alpha]}$ as a \emph{labelled polynomial}. 
By the definition of $\overline{\cdot}$, we always have $f^{[\alpha]} \in I \times \Fr$.
Furthermore, as done in~\cite{Sun11}, we denote by $f^{(\sigma)}$ a pair $(f,\sigma) \in K\<X> \times \MFr$ such that there exists $f^{[\alpha]} \in I^{[\Sigma]}$ with $\s(\alpha) = \sigma$.
We call such a pair a \emph{signature polynomial}.
Additionally, we denote by $I^{[\Sigma]}$ and $I^{(\Sigma)}$ the set of all labelled polynomials, respectively, the set of all signature polynomials, that is,
\begin{align*}
	I^{[\Sigma]} &\coloneqq \{f^{[\alpha]} \mid \alpha \in \Fr, f = \overline \alpha\} \subseteq I \times \Fr;\\
	I^{(\Sigma)} &\coloneqq \{f^{(\sigma)}\mid \exists f^{[\alpha]} \in I^{[\Sigma]}, \s(\alpha) = \sigma\}  \subseteq I \times M(\Sigma).
\end{align*}

\begin{remark}
We note that different families of generators of the same ideal $I \subseteq K\<X>$ always lead to different sets $I^{[\Sigma]}$.
To be more precise, given $I^{[\Sigma]}$, the family of generators of $I$ used in the construction can be recovered: $f_{i}$ is the polynomial part of the element $f_{i}^{[\varepsilon_{i}]}$ in $I^{[\Sigma]}$.

It is still true that different families of generators of $I$ can lead to different sets $I^{(\Sigma)}$, but it is not necessarily the case.
So to be precise, we should only speak of the set of signature polynomials $I^{(\Sigma)}$ \emph{w.r.t.\ the family of generators $(f_1,\dots,f_{r})$}.
However, whenever the generators $f_{1},\dots,f_{r}$ are clear from the context, we shall omit this part and only speak of $I^{(\Sigma)}$. 
\end{remark}

\chnote{Added this paragraph to make distinction clearer}
The motivation for using the notation $f^{[\alpha]}$ (resp.\ $f^{(\sigma)}$) is that the polynomial $f$, rather than the module element $\alpha$ (resp.\ the module monomial $\sigma$), is the main object of interest.
By abuse of language, we call $\s(\alpha)$ (resp.\ $\sigma$) the signature of the labelled polynomial $f^{[\alpha]}$ (resp.\ of the signature polynomial $f^{(\sigma)}$) and we denote it by $\s(f^{[\alpha]})$ (resp.~$\s(f^{(\sigma)})$).

The reason for introducing both labelled polynomials and signature polynomials is that the former allow to present the theory in a simpler fashion and lead to simpler proofs.
However, in an actual implementation of a signature-based algorithm, keeping track of the full module representation stored in a labelled polynomial causes a significant overhead in terms of memory consumption and overall computation time. 
Fortunately, we will see that all theoretical results only depend on information encoded in signature polynomials.
Consequently, when implementing a signature-based algorithm, one would only work with signature polynomials. 
This reduces the computational overhead. 
Additionally, we note that the reconstruction techniques discussed in Section~\ref{sec:reconstruction} allow to efficiently recover all information encoded in labelled polynomials from signature polynomials. 

In what follows, we will sometimes work with sets of polynomials, sometimes with sets of labelled polynomials and sometimes with sets of signature polynomials.
To be able to better distinguish between these cases, we denote subsets of $K\<X>$ by capital letters (e.g.\ $G \subseteq K\<X>$) and subsets of 
$I^{[\Sigma]}$ and $I^{(\Sigma)}$ by capital letters with the additional exponent  ${}^{[\Sigma]}$ and ${}^{(\Sigma)}$ respectively (e.g.~$G^{[\Sigma]} \subseteq I^{[\Sigma]}$ and $G^{(\Sigma)} \subseteq I^{(\Sigma)}$).

Computations in $I^{[\Sigma]}$ can be defined naturally. 
In particular, for $f^{[\alpha]}, g^{[\beta]}\in I^{[\Sigma]}$, $c \in K$ and $a,b \in \<X>$ we have
\begin{itemize}
	\item $f^{[\alpha]} + g^{[\beta]} = (f+g)^{[\alpha + \beta]}$, and
	\item $caf^{[\alpha]}b = (cafb)^{[ca\alpha b]}.$
\end{itemize}
With these operations the set $I^{[\Sigma]}$ becomes a $K\<X>$-bimodule.
Note that $f^{[\alpha]} = g^{[\beta]}$ if and only if $f = g$ and $\alpha = \beta$.



We call an element $\sigma \in \Fr$ a \emph{syzygy} if $\overline \sigma = 0$.
This corresponds to the element $0^{[\sigma]} \in I^{[\Sigma]}$.
The set of all syzygies of $f_1,\dots,f_r$ is denoted by $\Syz(f_1,\dots,f_r)$ and forms a $K\<X>$-submodule of $\Fr$.
For all labelled polynomials $f^{[\alpha]}, g^{[\beta]} \in I^{[\Sigma]}$ and any monomial $m \in \<X>$ we obtain with
$\sigma = \alpha m g -  f m \beta$ a so-called \emph{trivial syzygy} between $f^{[\alpha]}$ and $g^{[\beta]}$.

In order to discuss signature Gr\"obner bases, we need to adapt the notion of polynomial reduction to labelled polynomials.
This leads to the following definition of \emph{$\s$-reduction}.

\begin{definition}
Let $f^{[\alpha]}, f'^{[\alpha']}, g^{[\gamma]} \in I^{[\Sigma]}$ with $\alpha, g \neq 0$.
We say that $f^{[\alpha]}$ \emph{$\s$-reduces} to $f'^{[\alpha']}$ by $g^{[\gamma]}$ if there exist $a,b \in \<X>$ such that
\begin{itemize}
	\item $a \lm(g) b \in \supp(f)$, 
	\item $\s(a \gamma b) \preceq \s(\alpha)$, and
	\item $f'^{[\alpha']} = f^{[\alpha]} - \frac{\text{coeff}(f,a \lm(g) b)}{\lc(g)}a g^{[\gamma]} b$.
\end{itemize}
In this case, we write 
$f^{[\alpha]} \rightarrow_{g^{[\gamma]}} f'^{[\alpha']}$.

\tvnote{Moved those definitions here (and fixed them).}
If $a\lm(g)b = \lm(f)$, then the $\s$-reduction is called a \emph{top $\s$-reduction}.
Otherwise it is called a \emph{tail $\s$-reduction}.
The $\s$-reduction is called \emph{regular} if $\s(a\gamma b) \prec \s(\alpha)$.
Otherwise, that is, if $\s(a\gamma b) = \s(\alpha)$, it is called \emph{singular}.
\end{definition}

\tvnote{Added remarks on the operations using only the signatures}
\begin{remark}
In terms of polynomials, the first condition means that we can do usual polynomial reduction.
This implies that either $f' = 0$ or $\lm(f') \preceq \lm(f)$.
The second condition ensures that this inequality also transfers over to $\Fr$, i.e., that either $\alpha' = 0$ or $\s(\alpha') \preceq \s(\alpha)$.
Furthermore, if the $\s$-reduction is regular, then $\s(\alpha) = \s(\alpha')$.

Those observations allow to generalize the definition to signature polynomials.
More precisely, given $f^{(\alpha)}, g^{(\gamma)} \in I^{(\Sigma)}$, it is possible to test whether $f^{(\alpha)}$ is $\s$-reducible by $g^{(\gamma)}$.
Furthermore, if the $\s$-reduction is regular, it is possible to compute its remainder as a signature polynomial.
\end{remark}

A set of labelled polynomials $G^{[\Sigma]} \subseteq I^{[\Sigma]}$ induces a reduction relation $\rightarrow_{G^{[\Sigma]}} \subseteq I^{[\Sigma]} \times I^{[\Sigma]}$ by defining $f^{[\alpha]} \rightarrow_{G^{[\Sigma]}} f'^{[\alpha']}$ if there exists $g^{[\gamma]}\in G^{[\Sigma]}$ such that $f^{[\alpha]} \rightarrow_{g^{[\gamma]}} f'^{[\alpha']}$.
Furthermore, as in the case of polynomials, we denote by $\overset{*}{\rightarrow}_{G^{[\Sigma]}}$ the reflexive, transitive closure of $\rightarrow_{G^{[\Sigma]}}$.

If $f^{[\alpha]}$ $\s$-reduces to some $f'^{[\alpha']}$ with $f' = 0$, or in other words if $\alpha'$ is a syzygy, we say that $f^{[\alpha]}$ \emph{$\s$-reduces to zero}.

We capture some useful facts about $\s$-reduction that will be needed later.
\tvnote{Moved this down}
This first lemma follows immediately from the definition.
\begin{lemma}\label{lemma s-reduction}
If $f^{[\alpha]} \overset{*}{\rightarrow}_{G^{[\Sigma]}} f'^{[\alpha']}$, then $f \overset{*}\rightarrow_G f'$, where $G =  \{g \mid g^{[\gamma]} \in G^{[\Sigma]}\}$.
\end{lemma}

\tvnote{Added this}
The following lemma appears to be folklore, it is for example used in the proof of \cite[Lemma~9]{roune2012practical}. We include a proof for completeness.
\begin{lemma}\label{lemma top s-reducible sum}
Let $G^{[\Sigma]} \subseteq I^{[\Sigma]}$ and let $f^{[\alpha]}, g^{[\beta]} \in I^{[\Sigma]}$ be both top $\s$-reducible by $G^{[\Sigma]}$ with $\s(\alpha) \preceq \s(\alpha + \beta)$ and $\s(\beta) \preceq \s(\alpha + \beta)$.
Then, $f^{[\alpha]} + g^{[\beta]}$ is also top $\s$-reducible by $G^{[\Sigma]}$ or $\lt(f) + \lt(g) = 0$.
\end{lemma}
\begin{proof}
  Assume that $\lt(f) + \lt(g) \neq 0$.
  Then, $\lm( f + g) = \max \{ \lm(f), \lm(g)\}$.
  W.l.o.g.\ assume that $\lm(f) \succeq \lm(g)$, so that $\lm( f + g) = \lm(f)$.
  Since $f^{[\alpha]}$ is top $\s$-reducible by $G^{[\Sigma]}$ and $\s(\alpha) \preceq \s(\alpha + \beta)$, the labelled polynomial $f^{[\alpha]} + g^{[\beta]}$ is also top $\s$-reducible by $G^{[\Sigma]}$.
  In fact, the same element from $G^{[\Sigma]}$ can be used to top $\s$-reduce both $f^{[\alpha]}$ and $f^{[\alpha]} + g^{[\beta]}$.
\end{proof}

The outcome of classical polynomial reduction depends on more than just the leading term of the polynomial that is reduced.
Polynomials which share the same leading term can still reduce to different elements.
In case of regular $\s$-reductions, certain assumptions on the set of reducers $G^{[\Sigma]}$ imply that 
all labelled polynomials with the same signature term yield the same regular $\s$\nobreakdash-reduced result.
This fact is captured in the following lemma which is a noncommutative analogue of~\cite[Lemma~2]{roune2012practical}.

\begin{lemma}\label{lemma same signature}
Let $f^{[\alpha]}, g^{[\beta]} \in I^{[\Sigma]}$ be such that $\st(\alpha) = \st(\beta)$.
Furthermore, assume that all $u^{[\mu]} \in  I^{[\Sigma]}$ with $\s(\mu) \prec \s(\alpha)$ $\s$-reduce to zero by $G^{[\Sigma]}$.
Then, the following hold.
\begin{itemize}
  \item If $f^{[\alpha]}$ and $g^{[\beta]}$ are regular $\s$-reduced, then $f=g$.
  \item If $f^{[\alpha]}$ and $g^{[\beta]}$ are regular top $\s$-reduced, then $\lt(f)=\lt(g)$.
\end{itemize}
\end{lemma}
\begin{proof}
  The proof of the commutative version of this lemma~\cite[Lemma~2]{roune2012practical} carries over to the noncommutative setting.
\end{proof}

%

Using the notion of $\s$-reduction, we now define a \emph{labelled Gr\"obner basis} of the module $I^{[\Sigma]}$.

\begin{definition}\label{def s-GB}
A set $G^{[\Sigma]} \subseteq I^{[\Sigma]}$ is a \emph{labelled Gr\"obner basis of $I^{[\Sigma]}$ up to signature $\sigma \in \MFr$} if all $f^{[\alpha]} \in I^{[\Sigma]}$ with $\s(\alpha) \prec \sigma$ $\s$-reduce to zero by $G^{[\Sigma]}$. 
Furthermore, $G^{[\Sigma]}$ is a \emph{labelled Gr\"obner basis of $I^{[\Sigma]}$} if all $f^{[\alpha]} \in I^{[\Sigma]}$ $\s$-reduce to zero by $G^{[\Sigma]}$. 
\end{definition}

\begin{remark}
Since different families of generators of the ideal $I$ lead to different modules $I^{[\Sigma]}$, they also lead to different labelled Gr\"obner bases (see also Example~\ref{ex infinite s-GB}).
\end{remark}

\chnote{Moved definition of sig GB up here}
We recall that we present all the relevant theory for our signature-based algorithm in terms of labelled polynomials keeping in mind that in an actual implementation one would work only with signature polynomials. 
Consequently, such an implementation would not compute a labelled Gr\"obner basis but instead a signature Gr\"obner basis as defined below.

\begin{definition}\label{def sig-labelled GB}
A set $G^{(\Sigma)} \subseteq I^{(\Sigma)}$ is a \emph{signature Gr\"obner basis of $I^{[\Sigma]}$ (up to signature $\sigma \in \MFr$)}
if there exists a labelled Gr\"obner basis $G^{[\Sigma]} \subseteq I^{[\Sigma]}$ (up to signature $\sigma$) such that 
\begin{align}\label{eq def sig labelled}
	G^{(\Sigma)} = \{g^{(\s(\gamma))} \mid g^{[\gamma]} \in G^{[\Sigma]}\}.
\end{align}
\end{definition}

\chnote{Added this paragraph.}
In the theoretical sections, which are this section and the next one,
we focus on the notion of labelled Gr\"obner bases and present all theoretical results in terms of this concept.
However, one should always keep in mind that all relevant results also transfer over to signature Gr\"obner bases.
In Section~\ref{sec:reconstruction}, we then shift our focus to a more application-oriented point of view and present our signature-based algorithm in terms of signature polynomials.
Additionally, we show how to reconstruct a labelled Gr\"obner basis from a signature Gr\"obner basis.

It follows directly from the definition that a labelled Gr\"obner basis of $I^{[\Sigma]}$ is by no means unique. 
In fact, if $G^{[\Sigma]} \subseteq I^{[\Sigma]}$ is a labelled Gr\"obner basis of $I^{[\Sigma]}$, then so is $G^{[\Sigma]} \cup \{ f^{[\alpha]} \}$ for every $f^{[\alpha]} \in I^{[\Sigma]}$.
Furthermore, the set $I^{[\Sigma]}$ is always a labelled Gr\"obner basis of $I^{[\Sigma]}$.
Thus, we can immediately deduce the following corollary.

\begin{corollary}
For every finite family of generators of an ideal $I \subseteq K\<X>$, the module $I^{[\Sigma]}$ has a (possibly infinite) labelled Gr\"obner basis.
\end{corollary}

We also provide the following equivalent characterization of labelled Gr\"obner bases, which will turn out to be useful later.

\begin{lemma}\label{lemma equiv characterisation s-GB} 
A set $G^{[\Sigma]} \subseteq I^{[\Sigma]}$ is a labelled Gr\"obner basis of $I^{[\Sigma]}$ up to signature $\sigma \in \MFr$ if and only if every $f^{[\alpha]}\in I^{[\Sigma]}$ with $\s(\alpha) \prec \sigma$ is top $\s$-reducible by $G^{[\Sigma]}$.
Furthermore, $G^{[\Sigma]}$ is a labelled Gr\"obner basis of $I^{[\Sigma]}$ if and only if every $f^{[\alpha]}\in I^{[\Sigma]}$ is top $\s$-reducible by $G^{[\Sigma]}$.
\end{lemma}
\begin{proof}
The proof is an adaptation of the proof of the commutative version of this statement without signatures~\cite[Theorem~5.35]{Becker-1993}.
One only has to replace commutative polynomials by labelled (noncommutative) polynomials and polynomial reduction by $\s$\nobreakdash-reduction.
\end{proof}
\tvnote{Reference with signatures?}

The following proposition relates labelled Gr\"obner bases to Gr\"obner bases and is an immediate consequence of Lemma~\ref{lemma s-reduction}.
\begin{proposition}
Let $G^{[\Sigma]} \subseteq I^{[\Sigma]}$ be a labelled Gr\"obner basis of $I^{[\Sigma]}$.
Then, $\{g \mid g^{[\gamma]} \in G^{[\Sigma]}\} \subseteq K\<X>$ is a Gr\"obner basis of $I$.
\end{proposition}

Although a labelled Gr\"obner basis $G^{[\Sigma]}$ of $I^{[\Sigma]}$ is not unique in general, we can demand certain additional properties from $G^{[\Sigma]}$ in order to at least obtain a labelled Gr\"obner basis which is as small as possible.
We call such a labelled Gr\"obner basis a \emph{minimal} labelled Gr\"obner basis of $I^{[\Sigma]}$.
\begin{definition}
Let $G^{[\Sigma]} \subseteq I^{[\Sigma]}$ be a labelled Gr\"obner basis of $I^{[\Sigma]}$ (up to signature $\sigma \in \MFr$).
Then, $G^{[\Sigma]}$ is called a \emph{minimal} labelled Gr\"obner basis of $I^{[\Sigma]}$ (up to signature $\sigma$) if no $g^{[\gamma]} \in G^{[\Sigma]}$ can be top $\s$-reduced by $G^{[\Sigma]} \setminus \{g^{[\gamma]}\}$.
\end{definition}

We also extend this definition to signature Gr\"obner bases.

\begin{definition}
Let $G^{(\Sigma)} \subseteq I^{(\Sigma)}$ be a signature Gr\"obner basis of $I^{[\Sigma]}$ (up to signature $\sigma \in \MFr$).
Furthermore, let $G^{[\Sigma]} \subseteq I^{[\Sigma]}$ be a labelled Gr\"obner basis (up to signature $\sigma$) such that 
\begin{align}
	G^{(\Sigma)} = \{g^{(\s(\gamma))} \mid g^{[\gamma]} \in G^{[\Sigma]}\}.
\end{align}
Then, $G^{(\Sigma)}$ is called \emph{minimal} if $G^{[\Sigma]}$ is minimal.
\end{definition}

A minimal labelled Gr\"obner basis is minimal in the following sense.

\begin{proposition}\label{prop minimal s-GB}
Let $G^{[\Sigma]}\subseteq I^{[\Sigma]}$ be a minimal labelled Gr\"obner basis and $H^{[\Sigma]}\subseteq I^{[\Sigma]}$ be a labelled Gr\"obner basis of $I^{[\Sigma]}$.
Then, for every $g^{[\gamma]} \in G^{[\Sigma]}$ there exists $h^{[\delta]} \in H^{[\Sigma]}$ such that
\[
	 \lm(g) = \lm(h) \quad \text{ and }\quad  \s(\gamma) = \s(\delta).
\]
\end{proposition}

\begin{proof}
Let $g^{[\gamma]} \in G^{[\Sigma]}$. 
Since $H^{[\Sigma]}$ is a labelled Gr\"obner basis, $g^{[\gamma]}$ can be $\s$-reduced to zero by $H^{[\Sigma]}$.
In particular, this means that $g^{[\gamma]}$ is top $\s$-reducible by $H^{[\Sigma]}$, that is, there exist $h^{[\delta]} \in H^{[\Sigma]}$ and $a,b \in \<X>$ such that
\[
	\lm(g) = \lm(ahb) \quad \text{ and } \quad \s(\gamma) \succeq \s(a\delta b).
\]
Similarly, since $G^{[\Sigma]}$ is also a labelled Gr\"obner basis, there exist $g'^{[\gamma']} \in G^{[\Sigma]}$ and $a',b' \in \<X>$ such that
\[
	 \lm(h) = \lm(a'g'b') \quad \text{ and }\quad  \s(\delta) \succeq \s(a' \gamma' b').
\]
Combining these two statements yields
\[
	\lm(g) = \lm(ahb) = \lm(aa' g' b' b) \quad \text{ and }\quad \s(\gamma) \succeq \s(a \delta b) \succeq \s(a a' \gamma' b' b).
\]
Now, if $g^{[\gamma]} \neq g'^{[\gamma']}$, then $g'^{[\gamma']}$ could be used to top $\s$-reduce $g^{[\gamma]}$ but this is a contradiction to the fact that $G^{[\Sigma]}$ is minimal.
So, $g^{[\gamma]} = g'^{[\gamma']}$, which implies that $a = a' = b = b' = 1$, and therefore,
\[
	\lm(g) = \lm(h) \quad \text{ and } \quad \s(\gamma) = \s(\delta).\qedhere
\]
\end{proof}

We note that, starting with a finite family of generators $(f_1,\dots,f_r) \in K\<X>^r$ of an ideal $I \subseteq K\<X>$, the module $I^{[\Sigma]}$ always has a minimal labelled Gr\"obner basis, which is finite if and only if $I^{[\Sigma]}$ has a finite labelled Gr\"obner basis w.r.t.\ the family of generators $f_1,\dots,f_r$.
This follows from the following proposition, which tells us that we can obtain a minimal labelled Gr\"obner basis from a labelled Gr\"obner basis by removing all elements that are top $\s$-reducible.

\begin{proposition}\label{prop:min-GB-removing}
Let $G^{[\Sigma]} \subseteq I^{[\Sigma]}$ be a labelled Gr\"obner basis of $I^{[\Sigma]}$ such that there exists $g^{[\gamma]} \in G^{[\Sigma]}$ which is top $\s$-reducible by $G^{[\Sigma]}\setminus\{g^{[\gamma]}\}$.
Then, $G^{[\Sigma]}\setminus\{g^{[\gamma]}\}$ is also a labelled Gr\"obner basis of $I^{[\Sigma]}$.
\end{proposition}

\begin{proof}
If $g^{[\gamma]}$ is top $\s$-reducible by $G^{[\Sigma]}\setminus\{g^{[\gamma]}\}$, then there exist $g'^{[\gamma']} \in G^{[\Sigma]}\setminus\{g^{[\gamma]}\}$ and $a,b \in \<X>$ such that
$\lm(g) = \lm(ag'b)$ and $\s(\gamma) \succeq \s(a \gamma' b)$.
So, every element $f^{[\alpha]}\in I^{[\Sigma]}$ which is top $\s$\nobreakdash-reducible by $g^{[\gamma]}$ is also top $\s$-reducible by $g'^{[\gamma']}$.
Consequently, it follows from Lemma~\ref{lemma equiv characterisation s-GB} that $G^{[\Sigma]}\setminus\{g^{[\gamma]}\}$ is also a labelled Gr\"obner basis of $I^{[\Sigma]}$.
\end{proof}

\begin{corollary}\label{cor finite s-GB}
  The module $I^{[\Sigma]}$ has a finite labelled Gr\"obner basis if and only if $I^{[\Sigma]}$ has a finite minimal labelled Gr\"obner basis.
  Furthermore, all minimal labelled Gröbner bases of $I^{[\Sigma]}$ have the same cardinality.
\end{corollary}
\begin{proof}
  If $I^{[\Sigma]}$ has a finite labelled Gröbner basis $G^{[\Sigma]}$, then, applying Proposition~\ref{prop:min-GB-removing} repeatedly, $G^{[\Sigma]}$ contains a minimal labelled Gröbner basis of $I^{[\Sigma]}$ as a subset.
  The converse is clear.

  For the last statement, assume that $G_{1}^{[\Sigma]}$ and  $G_{2}^{[\Sigma]}$ are minimal labelled Gröbner bases of $I^{[\Sigma]}$.
  If they are both infinite, they have the same (countable) cardinality.
  Otherwise, assume that $G_{1}^{[\Sigma]}$ is finite.
  Define the sets $R_{1}$ and $R_{2}$ by
  \begin{equation*}
    R_{i} = \{ (\lm(g),\s(\gamma)) \mid g^{[\gamma]}\in G_{i}^{[\Sigma]}\}.
  \end{equation*}
  By Proposition~\ref{prop minimal s-GB}, $R_{1} \subseteq R_{2}$ and $R_{2} \subseteq R_{1}$, so $R_{1}=R_{2}$.
  We claim that the cardinality of $G_{1}^{[\Sigma]}$ and $G_{2}^{[\Sigma]}$ is equal to that of $R_{1}$.

  Indeed, assume that it is not the case.
  W.l.o.g.\ we can assume that $G_{1}^{[\Sigma]}$ is larger than $R_{1}$, so by the pigeonhole principle, there exist $g^{[\gamma]}$ and $h^{[\delta]}$ in $G_{1}^{[\Sigma]}$, distinct, such that $\lm(g)=\lm(h)$ and $\s(\gamma)=\s(\delta)$.
  Then by definition, $h^{[\delta]}$ is top $\s$-reducible by $g^{[\gamma]}$, which contradicts the minimality of $G_{1}^{[\Sigma]}$.
\end{proof}

However, we cannot expect $I^{[\Sigma]}$ to have a finite (minimal) labelled Gr\"obner basis for any finitely generated ideal $I \subseteq K\<X>$, as there are finitely generated ideals that simply do not have a finite Gr\"obner basis, and consequently, also no finite labelled Gr\"obner basis. 
Unfortunately, the condition that an ideal $I = (f_1,\dots,f_r)$ has a finite Gr\"obner basis is also not sufficient to ensure that $I^{[\Sigma]}$ has a finite labelled Gr\"obner basis w.r.t.\ the family of generators $f_1,\dots,f_r$.

\begin{example}[label=infinite-sgb]
  \label{ex infinite s-GB}
  We give an example of an ideal with a finite Gröbner basis, but no finite labelled Gröbner basis.
  The construction and the proof of the claims rely on notions introduced in Section~\ref{sec:comp-sign-grobn}, and will be deferred until that point.
  
  Let $K$ be a field and $X = \{x,y\}$.
  We consider the ideal $I = (f_1,f_2,f_3)$ with
  \[
    f_1 = xyx - xy, \qquad f_2 = yxy, \qquad f_3 = xyy - xxy \in K\<X>.
  \]
  We equip $K\<X>$ with $\preceq_{\textup{deglex}}$ where we order the indeterminates as $x \prec_{\textup{lex}} y$ and we use $\preceq_{\textbf{top}}$ as a module ordering.
  Then, $G = \{f_1,f_2,f_3,f_4\}$, where  $f_4 = xxy$, is a Gr\"obner basis of $I$.
  
  A minimal labelled Gröbner basis of $I^{[\Sigma]}$, w.r.t.\ the family of generators $f_1,f_2,f_3$, is given by
  \[
    G^{[\Sigma]} = \{f_1^{[\varepsilon_1]}, f_2^{[\varepsilon_2]}, f_3^{[\varepsilon_3]}, f_4^{[\alpha]} \} \cup \{g_n^{[\gamma_n]} \mid n \geq 0\},
  \]
  with $g_n = yx^{n+2}y$ and certain $\alpha, \gamma_n \in \Fr$ such that $\s(\alpha) = \varepsilon_1 y$ and $\s(\gamma_n) = y\varepsilon_3y^n$.
  
  So, Corollary~\ref{cor finite s-GB} implies that $I^{[\Sigma]}$ does not have a finite labelled Gr\"obner basis.
  The obstruction to having a finite labelled Gröbner basis is that $f_4^{[\alpha]}$ cannot be used to $\s$-reduce any of the elements $f_3^{[\varepsilon_{3}]}$ or $g_{n}^{[\gamma_{n}]}$.
  Indeed, since $\s(\alpha) \succ \varepsilon_{3}$ and $yx^{n}\s(\alpha) \succ \s(\gamma_{n})$, the reductions would cause the signatures to increase, and would not be $\s$-reductions.


If instead we consider the family of generators $f_1^{[\varepsilon_1]},f_2^{[\varepsilon_2]},f_3^{[\varepsilon_3]},f_4^{[\varepsilon_4]}$ generating the module $I^{[\Sigma']}$, where $\Sigma' = (K\<X> \otimes K\<X>)^4$ denotes the free $K\<X>$-bimodule of rank 4, then the finite set
  \[
  	G^{[\Sigma']} = \{f_1^{[\varepsilon_1]}, f_2^{[\varepsilon_2]}, f_3^{[\varepsilon_3]}, f_4^{[\varepsilon_4]} \} \cup \{yxxy^{[\gamma_0]} \}
\]
with $\s(\gamma_0) =  y\varepsilon_3$ is a minimal labelled Gr\"obner basis of $I^{[\Sigma']}$.
The difference is that $f_{4}^{[\varepsilon_{4}]}$ has now signature $\varepsilon_4$ instead of $\varepsilon_1y$, which renders all $g_{n}^{[\gamma_{n}]}$ with $n \geq 1$ top $\s$-reducible by this element.
\end{example}
\begin{remark}
  One can compare this example with Example~\ref{ex infinite gb}, where we considered the principal ideal $(f_1)$.
  Adding $f_2$  and $f_3$ to the ideal allowed it to have a finite Gröbner basis.
  Note that neither $f_2$ nor $f_3$ lie in $(f_1)$.
  Indeed, having a finite Gröbner basis is a property of the ideal, independently of its generators (see e.g.~\cite[Sec.~6]{Mor94} for further information).

  By contrast, the polynomial $f_4$ lies in the ideal $I = (f_1,f_2,f_3)$, but adding $f_4$ as a generator allows the ideal to have a finite labelled Gröbner basis.
  Here, the ideal spanned is the same, but the underlying module is different:
  the module parts of the elements in $I^{[\Sigma]}$ (w.r.t. $f_1,f_2,f_3$) lie in the free bimodule $\Sigma$ of rank $3$, whereas the module parts of the elements in $I^{[\Sigma']}$ (w.r.t. $f_{1},f_2,f_3,f_4$) lie in the free bimodule $\Sigma'$ of rank $4$.
 This causes $I^{[\Sigma]}$ to only contain $f_4^{[\varepsilon_{1}y + \dots]}$ while $I^{[\Sigma']}$ additionally contains $f_4^{[\varepsilon_{4}]}$.
  In general, different choices of generators of $I$ can lead to drastically different module structures for $I^{[\Sigma]}$, and to different (minimal) labelled Gröbner bases.
\end{remark}

\section{Computation of labelled Gröbner bases}
\label{sec:comp-sign-grobn}

\subsection{Regular S-polynomials}\label{sec s-GB computation}

The objective of this section is to state an adaptation of the noncommutative version of Buchberger's algorithm to include signatures.
To this end, we need to adapt the notion of S-polynomials to the case of labelled polynomials.
We first extend the notion of \emph{ambiguities} from~\cite{Ber78} from polynomials to labelled polynomials.
Recall that we fixed an indexed family of generators $(f_1,\dots,f_r) \in K\<X>^r$ of an ideal $I = (f_1,\dots,f_r) \subseteq K\<X>$ as well as a monomial and a module ordering.

\begin{definition}
  \label{def:S-pol}
Let $G^{[\Sigma]} \subseteq I^{[\Sigma]}$ and let $f^{[\alpha]}, g^{[\beta]} \in G^{[\Sigma]}$ be such that $f,g \neq 0$.
If $\lm(f) = AB$ and $\lm(g) = BC$ for some words $A,B,C\in\<X>\setminus\{1\}$, then we call the tuple
\[
	a = (ABC,A,C,f^{[\alpha]}, g^{[\beta]})
\]
an \emph{overlap ambiguity} of $G^{[\Sigma]}$.
We define its \emph{S-polynomial} $\spol(a)$ to be
\[
	\spol(a) \coloneqq \frac{1}{\lc(f)}f^{[\alpha]}C -\frac{1}{\lc(g)}Ag^{[\beta]}.
\]
Similarly, if $f \neq g$, $\lm(f) = ABC$ and $\lm(g) = B$ for some words $A,B,C\in\<X>$, then we call the tuple
\[
	a = (ABC,A,C,f^{[\alpha]}, g^{[\beta]})
\]
an \emph{inclusion ambiguity} of $G^{[\Sigma]}$.
We define its \emph{S-polynomial} $\spol(a)$ to be
\[
	\spol(a) \coloneqq \frac{1}{\lc(f)}f^{[\alpha]} -\frac{1}{\lc(g)}Ag^{[\beta]}C.
\]

\end{definition}

Disregarding the module labelling in the definition recovers the usual constructions of the noncommutative version of Buchberger's algorithm.
We recall a few classical observations on this construction, which might be unfamiliar to a reader more used to the commutative case.
First, two labelled polynomials can have more than one ambiguity with each other. 
Furthermore, an element can also form overlap ambiguities with itself.

In the noncommutative case, following~\cite{Ber78}, we only form S-polynomials in the presence of an ambiguity.
This is an analogue of Buchberger's coprime criterion (sometimes called GCD criterion)~\cite[Cor.~5.8]{Mor94}.
Contrary to the commutative case, the criterion is embedded in the definition of an S-polynomial.
This ensures that two polynomials can only give rise to finitely many S-polynomials.
This is necessary to ensure that Buchberger's algorithm terminates whenever a finite Gröbner basis exists.

An ambiguity $a = (ABC,A,C,f^{[\alpha]}, g^{[\beta]})$ is called \emph{singular} if
\[
	\s(\alpha C) = \s(A\beta)
\]
in case that $a$ is an overlap ambiguity, respectively if
\[
	\s(\alpha) = \s(A\beta C) 
\]
 in case that $a$ is an inclusion ambiguity.
If an ambiguity is not singular, it is called \emph{regular}.
We call an S-polynomial \emph{regular} (resp. \emph{singular}), if the respective ambiguity is regular (resp. singular).
Similarly to the case of reductions, it is possible to compute regular S-polynomials of signature polynomials.

In the following, we collect some useful results about S-polynomials.
We start by relating the signature of a regular S-polynomial to the signatures of the two input elements.
This first proposition is an immediate consequence of the definition of a regular ambiguity.

\begin{proposition}\label{cor signature S-poly}
Let $a = (ABC,A,C,f^{[\alpha]}, g^{[\beta]})$ be a regular ambiguity of a set $G^{[\Sigma]} \subseteq I^{[\Sigma]}$.
Then,
\[
	\s(\spol(a)) \succeq \max \{\s(\alpha), \s(\beta)\}.
      \]
\end{proposition}

\begin{lemma}\label{lemma signature s-poly}
Let $G^{[\Sigma]} \subseteq I^{[\Sigma]}$ and $f^{[\alpha]} \in I^{[\Sigma]}$.
Assume that $f^{[\alpha]}$ is not top \mbox{$\s$-reducible} by $G^{[\Sigma]}$.
Then, any regular S-polynomial between $f^{[\alpha]}$ and any element in $G^{[\Sigma]} \cup \{f^{[\alpha]}\}$ has signature strictly larger than $\s(\alpha)$.
\end{lemma}

\begin{proof}
Let $a = (ABC,A,C,f_1^{[\alpha_1]}, f_2^{[\alpha_2]})$ be a regular ambiguity between $f^{[\alpha]}$ and some element $g^{[\gamma]} \in G^{[\Sigma]} \cup \{f^{[\alpha]}\}$, i.e.,
\[
	f_1^{[\alpha_1]} = f^{[\alpha]} \text{ and } f_2^{[\alpha_2]} = g^{[\gamma]} \qquad \text{ or } \qquad f_1^{[\alpha_1]} = g^{[\gamma]} \text{ and } f_2^{[\alpha_2]} =  f^{[\alpha]}.
\]
According to Proposition~\ref{cor signature S-poly}, we have $\s(\spol(a)) \succeq \s(\alpha)$.
Assume for contradiction that $\s(\spol(a)) = \s(\alpha)$.
Then, $a$ cannot be an overlap ambiguity as otherwise
\[
	\s(\alpha) = \s(\spol(a)) = \max \{\s(\alpha_1C),\s(A\alpha_2) \} \succ \max \{\s(\alpha_1),\s(\alpha_2) \} \succeq \s(\alpha),
\]
where the strict inequality follows from the fact that $A,C \neq 1$.
So, $a$ must be an inclusion ambiguity. 
Note that then $g^{[\gamma]} \neq f^{[\alpha]}$, as a labelled polynomial cannot form inclusion ambiguities with itself.
Therefore, $g^{[\gamma]} \in G^{[\Sigma]}$.
Furthermore, by definition of a regular S-polynomial we know
\[
	\s(\alpha) = \s(\spol(a)) = \max \{\s(\alpha_1),\s(A\alpha_2C) \} \succeq \s(A\alpha_2 C).
\]
Now, if $f_1^{[\alpha_1]} = f^{[\alpha]}$ (and consequently $f_2^{[\alpha_2]} = g^{[\gamma]}$), then $\s(\alpha) \succeq \s(A\gamma C)$.
Since also $\lm(f) = \lm(AgC)$, this shows that $f^{[\alpha]}$ is top $\s$-reducible by $g^{[\gamma]}$, which is a contradiction.
If $f_2^{[\alpha_2]} = f^{[\alpha]}$ (and consequently $f_1^{[\alpha_1]} = g^{[\gamma]}$), then $\s(\alpha) \succeq \s(A\alpha C)$.
This is only possible if $A =C = 1$.
So, $\lm(f) = \lm(AgC) = \lm(g)$ and $\s(\alpha) \succeq \s(\gamma)$, showing again that $f^{[\alpha]}$ is top $\s$-reducible by $g^{[\gamma]}$, which is a contradiction.
\end{proof}

\subsection{Characterization of labelled Gröbner bases}
\label{sec:char-sign-grobn}

As in the commutative case, the design and the proof of correctness of the algorithm will rely on a signature variant of Buchberger's characterization of Gröbner bases, stating that if all \emph{regular} S-polynomials $\s$-reduce to zero, then one has a labelled Gröbner basis.

A particularity of the noncommutative case is that we will need to handle trivial syzygies separately in the proof process.
This is because the noncommutative definition of S-polynomials (Definition~\ref{def:S-pol}) effectively contains the restrictions granted by Buchberger's coprime criterion, eliminating some trivial syzygies.
Without those restrictions, the algorithms (even without signatures) would rarely terminate, because the module of trivial syzygies is in general not finitely generated.

In order to prove the noncommutative characterization of labelled Gröbner bases, we prove several lemmas.
\chnote{Rephrased this paragraph a bit.}
Leaving aside the provisions for trivial syzygies, we first prove Lemma~\ref{lemma cases}, which states 
that, given two labelled polynomials with the same leading monomial and different signatures, one can find a regular S-polynomial whose signature divides the larger signature.
This lemma is technical and useful for the rest of the proofs.
It ensures that it suffices to consider regular S-polynomials.

Then, we prove Lemma~\ref{lemma regular top reduced S-polynomial}, which states that given one regular S-polynomial $q^{[\rho]}$, one can find another regular S-polynomial $p^{[\pi]}$ such that $\s(\pi)$ and $\s(\rho)$ have a common multiple in $\MFr$ and such that $p^{[\pi]}$ satisfies certain additional conditions concerning its $\s$-reducibility.
This allows us to prove Lemma~\ref{lemma S-polynomial p}, which makes the same statement, but starting from any polynomial.
This last lemma is a noncommutative analogue of Lemma 9 in~\cite[appendix]{roune2012practical}.
Just like in the commutative case, it is the cornerstone of the proof of the final Theorem~\ref{thm spairs}, which is the wanted characterization.

\begin{lemma}\label{lemma cases}
 \tvnote{Center figure}
Let $g_1^{[\gamma_1]}, g_2^{[\gamma_2]} \in I^{[\Sigma]}$ and let $a_1,a_2,b_1,b_2\in \<X>$ such that
\[
	\lm(a_1g_1b_1) = \lm(a_2g_2b_2) \quad \text{ and }\quad  \s(a_1 \gamma_1 b_1) \succ \s(a_2 \gamma_2 b_2).
\]
Then, there exist $p^{[\pi]} \in I^{[\Sigma]}$ and $a,b\in \<X>$ such that $\s(a\pi b) = \s(a_1\gamma_1 b_1)$ and such that one of the following conditions holds:
\begin{enumerate}
	\item $\pi$ is a trivial syzygy between $g_1^{[\gamma_1]}$ and $g_2^{[\gamma_2]}$;
	\item $p^{[\pi]}$ is a regular S-polynomial of $g_1^{[\gamma_1]}$ and $g_2^{[\gamma_2]}$ with $p = 0$ or $\lm(apb) \prec \lm(a_1g_1b_1)$;
\end{enumerate}
\end{lemma}

\begin{figure}
  \centering
    \begin{tikzpicture}[
  node distance = 0mm,
    narrowbox/.style = {draw, minimum size=1cm, outer sep = 0mm, on chain, text width=4cm, draw, thin, text centered, minimum height=\ht\strutbox+\dp\strutbox,
      inner ysep=0pt},
    mybrace/.style={decorate,decoration={brace,raise=2pt,amplitude=5pt}},
    mymirrorbrace/.style={decorate,decoration={brace,mirror,amplitude=5pt}}
    ]

    \begin{scope}[yshift=2.5cm,start chain=going right]
      \node[narrowbox, anchor=west, text width = 1.5cm] at (0,0) (A) {$A$};
      \node[anchor=west] at (-2,0) (leg1) {Case 1:}; 
      \node[narrowbox, text width = 1.5cm] (m1) {$\lm(g_1)$}; 
      \node[narrowbox, text width = 1.5cm] (B) {$B$}; 
      \node[narrowbox, text width = 1.5cm] (m2) {$\lm(g_2)$};
      \node[narrowbox, text width = 1.5cm] (C) {$C$};
      
      \draw [mybrace] ([yshift=3mm]A.west) --node[above=2mm]{$a_1$} ([yshift=3mm]A.east);
      \draw [mybrace] ([yshift=3mm]B.west) --node[above=2mm]{$b_1$} ([yshift=3mm]C.east);
      \draw [mymirrorbrace] ([yshift=-4mm]A.west) --node[below=2mm]{$a_2$} ([yshift=-4mm]B.east);
      \draw [mymirrorbrace] ([yshift=-4mm]C.west) --node[below=2mm]{$b_2$} ([yshift=-4mm]C.east);
    \end{scope}

    \begin{scope}[yshift=0.3cm,start chain=going right]
      \node[narrowbox, anchor=west, text width = 1.5cm] at (0,0) (A) {$A$};
      \node[narrowbox, text width = 1.5cm] (m1) {$\lm(g_2)$}; 
      \node[narrowbox, text width = 1.5cm] (B) {$B$}; 
      \node[narrowbox, text width = 1.5cm] (m2) {$\lm(g_1)$};
      \node[narrowbox, text width = 1.5cm] (C) {$C$};
      
      \draw [mymirrorbrace] ([yshift=-4mm]A.west) --node[below=2mm]{$a_2$} ([yshift=-4mm]A.east);
      \draw [mymirrorbrace] ([yshift=-4mm]B.west) --node[below=2mm]{$b_2$} ([yshift=-4mm]C.east);
      \draw [mybrace] ([yshift=3mm]A.west) --node[above=2mm]{$a_1$} ([yshift=3mm]B.east);
      \draw [mybrace] ([yshift=3mm]C.west) --node[above=2mm]{$b_1$} ([yshift=3mm]C.east);
    \end{scope}

    \begin{scope}[yshift=-2.1cm,start chain=going right]
      \node[narrowbox, anchor=west, text width = 2cm] at (0,-.55) (a2) {$a_2$};
      \path (leg1.west) |- (a2) node[pos=0.5,anchor=west, yshift=5.5mm] {Case 2:}; 
      \node[narrowbox, text width = 4cm] (g2) {$\lm(g_2)$}; 
      \node[narrowbox, text width = 2cm] (b2) {$b_2$}; 
      \node[narrowbox, text width = 1cm, above = 5.5mm of a2.east, anchor =  west] (A) {$A$}; 
      \node[narrowbox, text width = 1.53cm, right =  of A.east, anchor =  west] (g1) {$\lm(g_1)$}; 
      \node[narrowbox, text width = 1cm, above =  5.5mm of b2.west, anchor =  east] (C) {$C$}; 
      \node[text width = 1cm, above = 5.5mm of a2.west, anchor = west] (a1) {};
      \node[text width = 1cm, above = 5.5mm of b2.east, anchor = east] (b1) {};

      \draw [mybrace] ([yshift=3mm]a1.west) --node[above=2mm]{$a_1$} ([yshift=3mm]A.east);
      \draw [mybrace] ([yshift=3mm]C.west) --node[above=2mm]{$b_1$} ([yshift=3mm]b1.east);
     \end{scope}

    \begin{scope}[yshift=-4.1cm,start chain=going right]
      \node[narrowbox, anchor=west, text width = 2cm] at (0,0) (a1) {$a_1$};
      \path (leg1.west) |- (a1) node[pos=0.5,anchor=west] {Case 3:}; 
      
      \node[narrowbox, text width = 4cm] (g1) {$\lm(g_1)$}; 
      \node[narrowbox, text width = 2cm] (b1) {$b_1$}; 
      \node[narrowbox, text width = 1cm, below = 5.5mm of a1.east, anchor =  west] (A) {$A$}; 
      \node[narrowbox, text width = 1.53cm, right =  of A.east, anchor =  west] (g2) {$\lm(g_2)$}; 
      \node[narrowbox, text width = 1cm, below =  5.5mm of b1.west, anchor =  east] (C) {$C$}; 
      \node[text width = 1cm, below = 5.5mm of a1.west, anchor = west] (a2) {};
      \node[text width = 1cm, below = 5.5mm of b1.east, anchor = east] (b2) {};
      
      \draw [mymirrorbrace] ([yshift=-4mm]a2.west) --node[below=2mm]{$a_2$} ([yshift=-4mm]A.east);
      \draw [mymirrorbrace] ([yshift=-4mm]C.west) --node[below=2mm]{$b_2$} ([yshift=-4mm]b2.east);
    \end{scope}

    \begin{scope}[yshift=-7cm,start chain=going right]
      \node[narrowbox, anchor=west, text width = 2cm] (a1) at (0,0) {$a_1$};
      \path (leg1.west) |- (a1) node[pos=0.5,anchor=west] {Case 4:}; 
      
      \node[narrowbox, text width = 2.5cm] (f1) {$\lm(g_1)$}; 
      \node[narrowbox, text width = 1.5cm] (C) {$C$}; 
      \node[text width = 2cm, right = of C.east, anchor = west] (right) {};
      \node[narrowbox, text width = 1.5cm, below =  5.5mm of f1.west, anchor =  west] (A) {$A$}; 
      \node[narrowbox, text width = 2.5cm, right =  of A.east, anchor =  west] (g2) {$\lm(g_2)$}; 
      \node[narrowbox, text width = 2cm, right = of g2.east, anchor =  west] (b2) {$b_2$}; 
      
      \node[text width = 2cm, left = of A.west, anchor = east] (left) {};
      
      \draw [mymirrorbrace] ([yshift=-4mm]left.west) --node[below=2mm]{$a_2$} ([yshift=-4mm]A.east);
      \draw [mybrace] ([yshift=3mm]C.west) --node[above=2mm]{$b_1$} ([yshift=3mm]right.east);
    \end{scope}

    \begin{scope}[yshift=-9.8cm,start chain=going right]
      \node[text width = 2cm, anchor=west] at (0,0) (left) {};
      \path (leg1.west) |- (left) node[pos=0.5,anchor=west] {Case 5:}; 
      
      \node[narrowbox, text width = 1.5cm, right = of left.east] (A) {$A$};
      \node[narrowbox, text width = 2.5cm] (g1) {$\lm(g_1)$}; 
      \node[narrowbox, text width = 2cm] (b1) {$b_1$}; 
      \node[narrowbox, text width = 2.5cm, below =  5.5mm of A.west, anchor =  west] (g2) {$\lm(g_2)$}; 
      \node[narrowbox, text width = 2cm, left = of g2.west, anchor =  east] (a2) {$a_2$}; 
      \node[narrowbox, text width = 1.5cm, right =  of g2.east, anchor =  west] (C) {$C$}; 
      \node[text width = 2cm, right = of C.east, anchor = west] (right) {};
      
      \draw [mybrace] ([yshift=3mm]left.west) --node[above=2mm]{$a_1$} ([yshift=3mm]A.east);
      \draw [mymirrorbrace] ([yshift=-4mm]C.west) --node[below=2mm]{$b_2$} ([yshift=-4mm]right.east);
    \end{scope}

  \end{tikzpicture}

  \caption{Relative position of $\lm(g_{1})$ and $\lm(g_{2})$  in the proof of Lemma \ref{lemma cases}}
  \label{fig:case-div}
\end{figure}
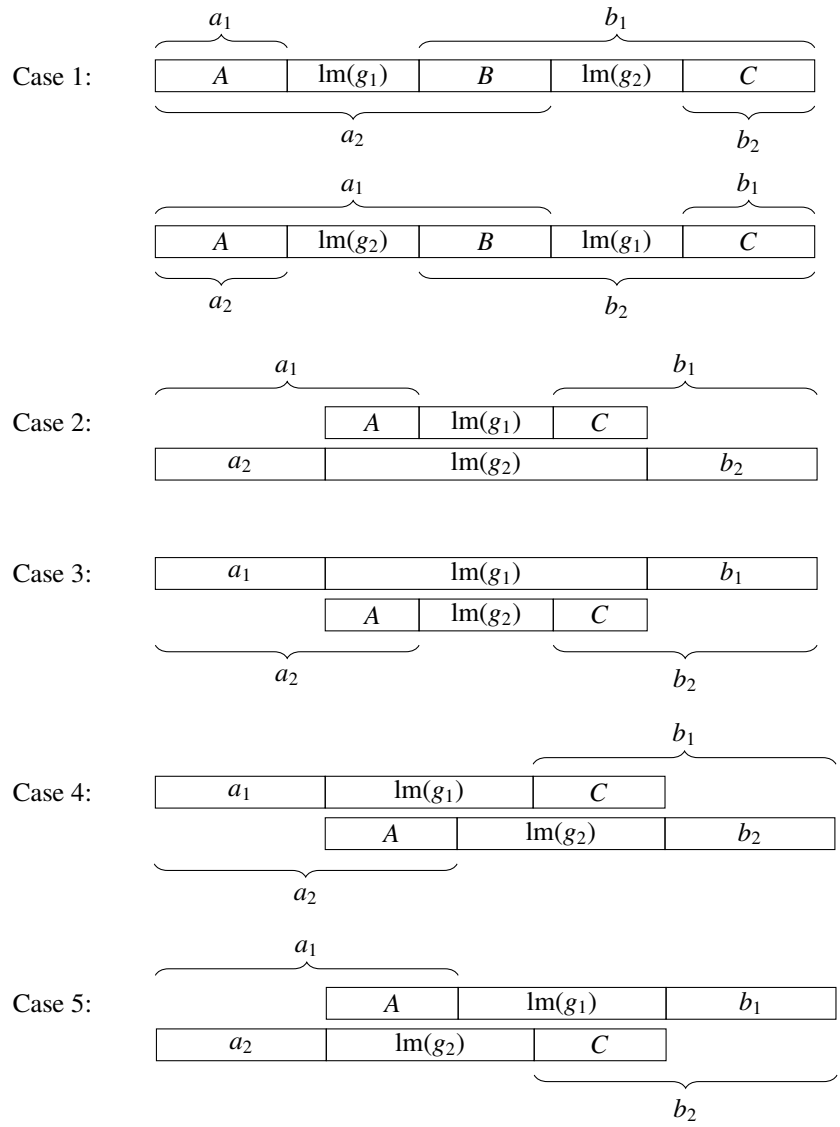

\begin{proof}
We distinguish between different cases (see Figure~\ref{fig:case-div}), depending on the position of $\lm(g_1)$ and $\lm(g_2)$ relative to each other in $W = \lm(a_1 g_1 b_1) = \lm(a_2 g_2 b_2)$.

\paragraph{Case 1: $\lm(g_1)$ is fully contained in $a_2$ or $b_2$}
In this case, $\lm(g_1)$ and $\lm(g_2)$ do not overlap in $W$.
We first consider the case where $\lm(g_1)$ is contained in $a_2$, i.e., where $W = A\lm(g_1)B \lm(g_2)C$ for some $A,B,C\in\<X>$.
We let
\[
	0^{[\pi]} = g_1^{[\gamma_1]}B g_2 - g_1B g_2^{[\gamma_2]}.
\]
Then $\pi$ is a trivial syzygy between $g_1^{[\gamma_1]}$ and $g_2^{[\gamma_2]}$.
To prove the assertion regarding the signatures, we note that $\s(\pi) = \s(\gamma_1 B g_2) \succ \s(g_1 B \gamma_2)$ as
$A\s(g_1B \gamma_2)C = \s(A \lm(g_1) B \gamma_2 C) =\s(a_2 \gamma_2 b_2) \prec \s(a_1\gamma_1b_1) = \s(A \gamma_1 B \lm(g_2) C) = A\s(\gamma_1 B g_2)C$.
Hence, by setting $a = A$ and $b = C$, we obtain
$\s(a\pi b) = \s(A\pi C) = \s(A \gamma_1 B g_2 C) = \s( a_1 \gamma_1b_1)$.

The other case, where $\lm(g_1)$ is contained in $b_2$, works along the same lines using the trivial syzygy $0^{[\pi]} = g_2^{[\gamma_2]}B g_1 - g_2B g_1^{[\gamma_1]}$.

\paragraph{Case 2: $\lm(g_1)$ is fully contained in $\lm(g_2)$}
In this case, there exists an inclusion ambiguity with S-polynomial 
\[
	p^{[\pi]} =  \frac{1}{\lc(g_2)}g_2^{[\gamma_2]} - \frac{1}{\lc(g_1)}Ag_1^{[\gamma_1]}C.
\]
with $\s(\pi) = \s(A\gamma_1C) \succ \s(\gamma_2)$ as
$a_2 \s(\gamma_2) b_2 = \s(a_2 \gamma_2 b_2) \prec \s(a_1 \gamma_1 b_1) = a_2 \s(A\gamma_1C) b_2$.
Hence, the S-polynomial is regular, and with $a = a_2$ and $b = b_2$ we have 
$\s(a\pi b) = \s(a_2A\gamma_1Cb_2) = \s(a_1\gamma_1b_1)$
and 
$\lm(apb) \prec \lm(ag_2 b) = \lm(a_2g_2b_2) = \lm(a_1g_1b_1)$,
in case $p \neq 0$.

\paragraph{Case 3: $\lm(g_2)$ is fully contained in $\lm(g_1)$}
In this case, there exists an inclusion ambiguity with S-polynomial 
\[
	p^{[\pi]} = \frac{1}{\lc(g_1)}g_1^{[\gamma_1]} - \frac{1}{\lc(g_2)}Ag_2^{[\gamma_2]}C
\]
with $\s(\pi) = \s(\gamma_1) \succ \s(A\gamma_2C)$ as 
$a_1 \s(A\gamma_2C) b_1 = \s(a_2 \gamma_2 b_2) \prec \s(a_1 \gamma_1 b_1) = a_1 \s(\gamma_1) b_1$.
Hence, the S-polynomial is regular, and with $a = a_1$ and $b = b_1$ we have
$\s(a \pi b) = \s(a_1 \gamma_1 b_1)$
and
$\lm(a p b) \prec \lm(a g_1 b) = \lm (a_1 g_1 b_1)$,
in case $p \neq 0$.

\paragraph{Case 4: $\lm(g_1)$ and $\lm(g_2)$ overlap but are not fully contained in one another and $\lm(g_1)$ begins before $\lm(g_2)$}
In this case, there exists an overlap ambiguity with S-polynomial 
\[
	p^{[\pi]} =  \frac{1}{\lc(g_1)}g_1^{[\gamma_1]}C - \frac{1}{\lc(g_2)}Ag_2^{[\gamma_2]}
\]
with $\s(\pi) = \s(\gamma_1C) \succ \s(A\gamma_2)$ as 
$a_1 \s(A\gamma_2) b_2 = \s(a_2\gamma_2 b_2) \prec \s(a_1 \gamma_1 b_1) = a_1 \s(\gamma_1C)b_2$.
Hence, the S-polynomial is regular, and with $a = a_1$ and $b = b_2$ we have
$\s(a \pi b) = \s(a_1 \gamma_1C b_2) = \s(a_1 \gamma_1 b_1)$
and
$\lm(apb) \prec \lm(ag_1Cb) = \lm(a_1g_1Cb_2) = \lm(a_1 g_1 b_1)$,
in case $p \neq 0$.

\paragraph{Case 5: $\lm(g_1)$ and $\lm(g_2)$ overlap but are not fully contained in one another and $\lm(g_1)$ begins after $\lm(g_2)$}
In this case, there exists an overlap ambiguity with S-polynomial 
\[
	p^{[\pi]} = \frac{1}{\lc(g_2)}g_2^{[\gamma_2]}C - \frac{1}{\lc(g_1)}Ag_1^{[\gamma_1]}
\]
with $\s(\pi) = \s(A\gamma_1) \succ \s(\gamma_2C)$ as 
$a_2 \s(\gamma_2C)b_1 = \s(a_2 \gamma_2 b_2) \prec \s(a_1 \gamma_1 b_1) = a_2\s(A\gamma_1)b_1$.
Hence, the S-polynomial is regular and with $a = a_2$ and $b = b_1$ we have
$\s(a \pi b) = \s(a_2 A\gamma_1 b_1) = \s(a_1\gamma_1b_1)$
and
$\lm(apb) \prec \lm(aAg_1b) = \lm(a_2Ag_1b_1) = \lm(a_1 g_1 b_1)$,
in case $p \neq 0$.
\end{proof}

\begin{lemma}\label{lemma regular top reduced S-polynomial}
Let $q^{[\rho]} \in I^{[\Sigma]}$ be a regular S-polynomial of $G^{[\Sigma]}$.
Furthermore, assume that there exist $a',b' \in \<X>$ such that all $u^{[\mu]} \in I^{[\Sigma]}$ with $\s(\mu) \prec \s(a'\rho b')$ $\s$-reduce to zero by $G^{[\Sigma]}$.
Then, there exist $p^{[\pi]}\in I^{[\Sigma]}$ and $a,b \in \<X>$ such that $\s(a\pi b) = \s(a'\rho b')$ and such that one of the following conditions holds:
\begin{enumerate}
	\item $\pi$ is a trivial syzygy between two elements in $G^{[\Sigma]}$;
	\item $p^{[\pi]}$ is a regular S-polynomial of $G^{[\Sigma]}$ and $ap'^{[\pi']}b$ is not regular top $\s$-reducible where $p'^{[\pi']}$ is the result of regular $\s$-reducing $p^{[\pi]}$;
\end{enumerate}
\end{lemma}

\begin{proof}
Let $q'^{[\rho']}$ be the result of regular $\s$-reducing $q^{[\rho]}$ by $G^{[\Sigma]}$.
If $a' q'^{[\rho']}b'$ is not regular top $\s$-reducible, then $q'^{[\rho']},a',b'$ are a suitable choice for $p^{[\pi]},a,b$.

Now assume that $a' q'^{[\rho']}b'$ is regular top $\s$-reducible.
This means that $a'b' \neq 1$, and therefore, that $\s(\rho') = \s(\rho) \prec a'\s(\rho)b' = \s(a' \rho b')$.
Hence, $q^{[\rho]}$ as well as $q'^{[\rho']}$ $\s$-reduce to zero by $G^{[\Sigma]}$.

We are now going to construct $p^{[\pi]} \in I^{[\Sigma]}$ such that there exist $a,b \in \<X>$ with $\s( a \pi b) = \s(a' \rho b')$ and 
\begin{enumerate}
	\item $\pi$ is a trivial syzygy between two elements in $G^{[\Sigma]}$, or
	\item $p^{[\pi]}$ is a regular S-polynomial with $p = 0$ or $\lm(apb) \prec \lm(a'qb')$.
\end{enumerate}

We are done if $\pi$ is a trivial syzygy or if $ap'^{[\pi']}b$ is not regular top $\s$-reducible where $p'^{[\pi']}$ is the result of regular $\s$-reducing $p^{[\pi]}$. 
Otherwise we can repeat this process to construct a third labelled polynomial with the same properties.
This process must terminate at some point since $\prec$ is a well-ordering.

We note that $q' \neq 0$ because otherwise $a' q'^{[\rho']}b'$ would not be regular top $\s$-reducible.
Then, since $q'^{[\rho']}$ is regular $\s$-reduced but $\s$-reduces to zero, $q'^{[\rho']}$ must be singular top $\s$-reducible.
Hence, there exist $g_1^{[\gamma_1]} \in G^{[\Sigma]}$ and $a_1,b_1 \in \<X>$ such that
\[
	\lm(a_1 g_1 b_1) = \lm(q') \quad \text{ and } \quad \s(a_1 \gamma_1 b_1) = \s(\rho') = \s(\rho).
\]
Furthermore, as $a'q'^{[\rho']}b'$ is regular top $\s$-reducible, there exist $g_2^{[\gamma_2]} \in G^{[\Sigma]}$ and $a_2,b_2 \in \<X>$ such that
\[
	\lm(a_2 g_2 b_2) = \lm(a'q'b') \quad \text{ and } \quad \s(a_2 \gamma_2 b_2) \prec \s(a' \rho' b') = \s(a'\rho b').
\]

To summarize, we have
\begin{align*}
	\lm(a'a_1 g_1 b_1b') &= \lm(a'q'b') = \lm(a_2 g_2 b_2),\\
	\s(a' a_1 \gamma_1 b_1 b') &= \s(a' \rho b') \succ \s(a_2 \gamma_2 b_2).
\end{align*}
Now, we apply Lemma~\ref{lemma cases} to $a'a_1g_1^{[\gamma_1]}b_1b'$ and $a_2g_2^{[\gamma_2]}b_2$,
which yields $a,b \in \<X>$ and $p^{[\pi]} \in I^{[\Sigma]}$ such that $\s(a \pi b) = \s(a' a_1 \gamma_1 b_1 b') = \s(a' \rho b')$ and such that one of the following conditions holds:
\begin{enumerate}
	\item $\pi$ is a trivial syzygy between $g_1^{[\gamma_1]}$ and $g_2^{[\gamma_2]}$;
	\item $p^{[\pi]}$ is a regular S-polynomial with $p = 0$ or $\lm(apb) \prec \lm(a' a_1 g_1 b_1 b') = \lm(a'q'b') \preceq \lm(a'qb')$;
\end{enumerate}
These are the desired $a,b$ and $p^{[\pi]}$.
\end{proof}

\begin{lemma}\label{lemma S-polynomial p}
Let $h^{[\delta]} \in I^{[\Sigma]}$ be top $\s$-reduced by $G^{[\Sigma]}$.
Assume that for all $i = 1,\dots,r$ with $\varepsilon_i \preceq \s(\delta)$ there exists $g_i^{[\gamma_i]} \in G^{[\Sigma]}$ with $\s(\gamma_i) = \varepsilon_i$.
Furthermore, assume that all $u^{[\mu]} \in  I^{[\Sigma]}$ with $\s(\mu) \prec \s(\delta)$ $\s$-reduce to zero by $G^{[\Sigma]}$.
Then, there exist $p^{[\pi]} \in I^{[\Sigma]}$ and $a,b\in \<X>$ such that $\s(\delta) = \s(a \pi b)$ and such that one of the following conditions holds:
\begin{enumerate}
	\item $\pi$ is a trivial syzygy between two elements in $G^{[\Sigma]}$;
	\item $p^{[\pi]}$ is a regular S-polynomial of $G^{[\Sigma]}$ and $ap'^{[\pi']}b$ is not regular top $\s$-reducible where $p'^{[\pi']}$ is the result of regular $\s$-reducing $p^{[\pi]}$;
\end{enumerate}
\end{lemma}

\begin{proof}
  The proof follows the same structure as that of~\cite[Lemma 9, appendix]{roune2012practical}:
  first, we combine leading terms to show that there exists $L^{[\lambda]} \in I^{[\Sigma]}$ with signature dividing that of $\delta$, and then, starting from that element, we construct $p^{[\pi]}$ as wanted.

\paragraph{Considering leading terms}
Let $\s(\delta) = a_1 \varepsilon_i b_1$ for some $1\leq i \leq r$ and $a_1, b_1 \in \<X>$.
By our assumption on $G^{[\Sigma]}$ there exists $g_1^{[\gamma_1]} \in G^{[\Sigma]}$ such that $\s(\gamma_1) = \varepsilon_i$.
Let $\lambda = c a_1 \gamma_1 b_1$ with $c =  \frac{\sigc(\delta)}{\sigc(\gamma_1)}$ and observe that $\st(\lambda) = \st(\delta)$.
Then, with $L = \overline{\lambda}$, the labelled polynomial $h^{[\delta]} - L^{[\lambda]}$ has a strictly smaller signature than $h^{[\delta]}$, so it $\s$-reduces to zero by $G^{[\Sigma]}$ and in particular it is top $\s$-reducible by $G^{[\Sigma]}$.
Clearly, so is $L^{[\lambda]} = c a_1 g_1^{[\gamma_1]} b_1$, as it is top $\s$-reducible by $g_1^{[\gamma_1]}$.
However, the sum $h^{[\delta]} = (h^{[\delta]} - L^{[\lambda]}) + L^{[\lambda]}$ is not top $\s$-reducible by assumption.
So, Lemma~\ref{lemma top s-reducible sum} yields that
\[
	\lt(h - L) + \lt(L) = 0.
\]
\paragraph{Constructing $p^{[\pi]}$}

Since $h^{[\delta]} - L^{[\lambda]}$ $\s$-reduces to zero by $G^{[\Sigma]}$, there exists a $g_2^{[\gamma_2]} \in G^{[\Sigma]}$ and $a_2, b_2 \in \<X>$ such that
\[
  \lm(a_2 g_2 b_2) = \lm(h - L) \quad \text{ and } \quad \s(a_2 \gamma_2 b_2) \preceq \s(\delta - \lambda) \prec \s(\delta).
\]

Hence, with $a_1,b_1,g_1^{[\gamma_1]}$ from above, we have
\begin{align*}
	\lm(a_1 g_1 b_1) &= \lm(L) = \lm(h - L) = \lm(a_2 g_2 b_2), \\
	\s(a_1 \gamma_1 b_1) &= \s(\delta) \succ \s(\delta - \lambda) \succeq \s(a_2 \gamma_2 b_2).
\end{align*}

Now, we apply Lemma~\ref{lemma cases} to $a_1g_1^{[\gamma_1]}b_1$ and $a_2g_2^{[\gamma_2]}b_2$,
which yields $a',b' \in \<X>$ and $q^{[\rho]} \in I^{[\Sigma]}$ such that $\s(a' \rho b') = \s(a_1 \gamma_1 b_1) = \s(\delta)$ and such that one of the following conditions holds:
\begin{enumerate}
	\item $\rho$ is a trivial syzygy between $g_1^{[\gamma_1]}$ and $g_2^{[\gamma_2]}$;
	\item $q^{[\rho]}$ is a regular S-polynomial with $q = 0$ or $\lm(a'qb') \prec \lm(a_1 g_1 b_1)$;
\end{enumerate}

If $\rho$ is a trivial syzygy, we can set $p^{[\pi]} = 0^{[\rho]}$, $a = a', b = b'$ and are done.
Otherwise, we note that by assumption all $u^{[\mu]} \in I^{[\Sigma]}$ with $\s(\mu) \prec \s(\delta) = \s(a' \rho b')$ $\s$-reduce to zero by $G^{[\Sigma]}$.
Hence, we can apply Lemma~\ref{lemma regular top reduced S-polynomial} to the regular S-polynomial $q^{[\rho]}$ which gives $p^{[\pi]} \in I^{[\Sigma]}$ and $a,b\in\<X>$ such that
$\s(a \pi b) = \s(a' \rho b') = \s(\delta)$ and such that one of the following conditions holds:
\begin{enumerate}
	\item $\pi$ is a trivial syzygy between two elements in $G^{[\Sigma]}$;
	\item $p^{[\pi]}$ is a regular S-polynomial of $G^{[\Sigma]}$ and $ap'^{[\pi']}b$ is not regular top $\s$-reducible where $p'^{[\pi']}$ is the result of regular $\s$-reducing $p^{[\pi]}$;
\end{enumerate}
These are the desired $a,b$ and $p^{[\pi]}$.
\end{proof}

We can now finally state and prove the following theorem.

\begin{theorem}\label{thm spairs}
Let $\sigma \in \MFr$ be a module monomial and let $G^{[\Sigma]} \subseteq I^{[\Sigma]}$ be such that for all $\varepsilon_i \prec \sigma$ there exists $g_i^{[\gamma_i]} \in G^{[\Sigma]}$ with $\s(\gamma_i) = \varepsilon_i$.
Assume that all regular S-polynomials $p^{[\pi]}$ of $G^{[\Sigma]}$ with $\s(\pi) \prec \sigma$ regular $\s$-reduce to some $p'^{[\pi']}$ by $G^{[\Sigma]}$ such that
$\pi'$ is a syzygy or $p'^{[\pi']}$ is singular top $\s$-reducible.
Then, $G^{[\Sigma]}$ is a labelled Gr\"obner basis of $I^{[\Sigma]}$ up to signature $\sigma$.
\end{theorem} 

\begin{remark}
\chnote{Remark for sig-redundant}
The notion of being singular top $\s$-reducible is equivalent to what is in the (commutative)
literature also called \emph{sig-redundant} (see~\cite{EP11}) and included in the concept of \emph{super top reductions} in~\cite{Gao-2015-new-framework-for}.
Additionally, a regular $\s$-reduced element being singular top $\s$-reducible corresponds to the notion of not being \emph{primitive $\s$-irreducible} in~\cite{AP11}.
\end{remark}

\begin{proof}[Proof of Theorem~\ref{thm spairs}]
Assume, for contradiction, that $G^{[\Sigma]}$ is not a labelled Gr\"obner basis of $I^{[\Sigma]}$ up to signature $\sigma$.
Then, there exists a labelled polynomial $h^{[\delta]} \in I^{[\Sigma]}$ with $\s(\delta) \prec \sigma$ which does not $\s$-reduce to zero by $G^{[\Sigma]}$.
W.l.o.g.\ we let $h^{[\delta]}$ be such that $\s(\delta)$ is minimal.
Furthermore, we can also assume that $h^{[\delta]}$ is top $\s$-reduced.
Then, according to Lemma~\ref{lemma S-polynomial p}, there exist $p^{[\pi]} \in I^{[\Sigma]}$ and $a,b \in \<X>$ such that $\s(a\pi b) = \s(\delta)$ and such that
\begin{enumerate}
	\item $\pi$ is a (trivial) syzygy (between two elements in $G^{[\Sigma]}$), or
	\item $p^{[\pi]}$ is a regular S-polynomial of $G^{[\Sigma]}$ and $ap'^{[\pi']}b$ is not regular top $\s$-reducible where $p'^{[\pi']}$ is the result of regular $\s$-reducing $p^{[\pi]}$.
\end{enumerate}
We distinguish between the two possible cases.

\paragraph{Case 1: $\pi$ is a syzygy}
We denote $0^{[\alpha]} = \frac{\sigc(\delta)}{\sigc(\pi)}a 0^{[\pi]} b$ and note that this element is (regular) top $\s$-reduced. 
Hence, since also $h^{[\delta]}$ is regular top \mbox{$\s$-reduced} and $\st(\delta) = \st(\alpha)$, Lemma~\ref{lemma same signature} yields that $\lt(h) = \lt(0) = 0$.
This implies that $h = 0$, which contradicts the assumption that $h^{[\delta]}$ does not $\s$-reduce to zero.

\paragraph{Case 2: $p^{[\pi]}$ is a regular S-polynomial of $G^{[\Sigma]}$}
Then, by construction we know that $ap'^{[\pi']}b$ is not regular top $\s$-reducible where $p'^{[\pi']}$ is the result of regular $\s$-reducing $p^{[\pi]}$.
By assumption, $\pi'$ is a syzygy or $p'^{[\pi']}$ is singular top $\s$-reducible.
In the first case, we can reuse the arguments from above to reach the same contradiction.
Hence, we can assume that $p' \neq 0$ and that $p'^{[\pi']}$ is singular top \mbox{$\s$-reducible}.
We denote $f^{[\alpha]} = \frac{\sigc(\delta)}{\sigc(\pi')}ap'^{[\pi']} b$ and note that this element is regular top $\s$-reduced since $ap'^{[\pi']}b$ is regular top $\s$-reduced.
Since also $h^{[\delta]}$ is regular top $\s$-reduced and $\st(\delta) = \st(\alpha)$, Lemma~\ref{lemma same signature} yields that $\lt(h) = \lt(f)$.
So, anything that top $\s$-reduces $ap'^{[\pi']}b$ also top $\s$\nobreakdash-reduces $h^{[\delta]}$.
We note that $ap'^{[\pi']}b$ is top $\s$-reducible as $p'^{[\pi']}$ is top $\s$-reducible.
Thus,  $h^{[\delta]}$ is top $\s$-reducible, which is a contradiction.
\end{proof}

\begin{example}[continues=infinite-sgb,label=infinite-calculations]
  Recall that in Example~\ref{ex infinite s-GB} we considered the ideal $I = (f_1,f_2,f_3) \subseteq K\<X>$ with
  \[
    f_1 = xyx - xy, \qquad f_2 = yxy, \qquad f_3 = xyy - xxy,
  \]
  over a field $K$ in the variables $X = \{x,y\}$.
  We also defined $f_4= x x y$ and used $\preceq_{\textup{deglex}}$, where $x \prec_{\textup{lex}} y$, as a monomial ordering and $\preceq_{\textbf{top}}$ as a module ordering.

  We claimed that a minimal labelled Gröbner basis of $I^{[\Sigma]}$, w.r.t.\ the family of generators $f_1,f_2,f_3$, is given by
  \[
    G^{[\Sigma]} = \{f_1^{[\varepsilon_1]}, f_2^{[\varepsilon_2]}, f_3^{[\varepsilon_3]}, f_4^{[\alpha]} \} \cup \{g_n^{[\gamma_n]} \mid n \geq 0\},
  \]
  with $g_n = yx^{n+2}y$ and certain $\alpha, \gamma_n \in \Fr$ such that $\s(\alpha) = \varepsilon_1 y$ and $\s(\gamma_n) = y\varepsilon_3y^n$.

  We now prove that $G^{[\Sigma]}$ is indeed a labelled Gröbner basis of $I^{[\Sigma]}$, using Theorem~\ref{thm spairs}.
  First, since $G^{[\Sigma]}$ contains $f_1^{[\varepsilon_1]}, f_2^{[\varepsilon_2]}$ and
  $f_3^{[\varepsilon_3]}$, the first hypothesis of the theorem is satisfied.
  Then, we verify that all regular S-polynomials top $\s$-reduce to $0$ or to a singular top $\s$-reducible element.
  It is a straightforward, if tedious, calculation, which is detailed in~\ref{appendix}.

  Then, we prove that $G^{[\Sigma]}$ is minimal.
  Looking at the leading terms, the only possible reductions would be using $f_4^{[\alpha]}$ to reduce $f_3^{[\varepsilon_3]}$ or $g_n^{[\gamma_n]}$.
  But as $\s(\alpha) \succ \varepsilon_3$ and $yx^n \s(\alpha) \succ \s(\gamma_n)$, those reductions would not be $\s$-reductions.
  So none of the elements of $G^{[\Sigma]}$ is $\s$-reducible modulo the others, and $G^{[\Sigma]}$ is a minimal labelled Gröbner basis.
  
  The claim that $I^{[\Sigma']}$ with the family of generators $f_1,f_2,f_3,f_4$ has minimal labelled Gröbner basis $G^{[\Sigma']} = \{f_1^{[\varepsilon_1]}, f_2^{[\varepsilon_2]}, f_3^{[\varepsilon_3]}, f_4^{[\varepsilon_4]} \} \cup \{yxxy^{[\gamma_0]} \}$ is proved along the same lines in~\ref{appendix}.
\end{example}

\subsection{Effective description of the module of syzygies}
\label{sec:algorithm}

Similarly to Buchberger's classical characterization of Gröbner bases, Theorem~\ref{thm spairs} allows us to state a first, non-optimized version of a signature-based algorithm for noncommutative polynomials, by ensuring that regular S-polynomials which are not trivial syzygies regular $\s$\nobreakdash-reduce to zero or to a singular top $\s$-reducible normal form.

The fact that we need to handle at least some trivial syzygies separately is a crucial difference to the commutative case: in the commutative case, Buchberger's coprime criterion and the F5 criterion allow to eliminate some (resp.\ all) trivial syzygies, but signature-based algorithms terminate even without the criteria.

By contrast, in the noncommutative case, the module of trivial syzygies is in general not finitely generated, which requires handling trivial syzygies separately.
In doing so, we are able to obtain an effective description of the module of syzygies.
More precisely, we state the following fact about syzygies.

\begin{lemma}\label{lemma syzygy module}
Let $\mu \in \Fr$ be a syzygy and let $G^{[\Sigma]} \subseteq I^{[\Sigma]}$ be a labelled Gr\"obner basis of $I^{[\Sigma]}$ up to signature $\s(\mu)$.
Then, there exist $p^{[\pi]} \in I^{[\Sigma]}$ and $a,b \in \<X>$ such that $\s(a \pi b) = \s(\mu)$ and such that one of the following conditions holds:
\begin{enumerate}
	\item $\pi$ is a trivial syzygy between two elements in $G^{[\Sigma]}$;
	\item $p^{[\pi]}$ is a regular S-polynomial of $G^{[\Sigma]}$ which regular $\s$-reduces to zero by $G^{[\Sigma]}$;
\end{enumerate}
\end{lemma}

\begin{proof}
Since $0^{[\mu]}$ is top $\s$-reduced by $G^{[\Sigma]}$, Lemma~\ref{lemma S-polynomial p} yields the existence of $p^{[\pi]} \in I^{[\Sigma]}$ and $a,b\in \<X>$ such that $\s(\mu) = \s(a\pi b)$ and such that one of the following conditions holds:
\begin{enumerate}
	\item $\pi$ is a trivial syzygy between two elements in $G^{[\Sigma]}$;
	\item $p^{[\pi]}$ is a regular S-polynomial of $G^{[\Sigma]}$ and $ap'^{[\pi']}b$ is not regular top $\s$-reducible where $p'^{[\pi']}$ is the result of regular $\s$-reducing $p^{[\pi]}$;
\end{enumerate}
If $\pi$ is a trivial syzygy, we are done.
Otherwise, since neither $0^{[\mu]}$ nor $cap'^{[\pi']}b$, with $c = \frac{\sigc(\mu)}{\sigc(\pi')}$, are regular top $\s$-reducible and $\st(\mu) = \st(ca \pi' b)$, Lemma~\ref{lemma same signature} yields that $\lt(cap'b) = \lt(0) = 0$, and consequently, also $p' = 0$.
\end{proof}

Lemma~\ref{lemma syzygy module} allows us to describe more precisely the syzygy module  $S=\Syz(f_1,\dots,f_r)$.
Consider the set $\Htriv$ of trivial trivial syzygies of $G^{[\Sigma]}$
\begin{equation*}
  \Htriv = \left\{\gamma_1 m g_2 - g_1 m \gamma_2 \mid g_1^{[\gamma_1]}, g_2^{[\gamma_2]} \in G^{[\Sigma]}, m \in \<X>\right\}.
\end{equation*}
Note that the set of signatures of $\Htriv$ contains all the elements of the form
\begin{equation}
  \label{eq:1}
  \max \{\s(\gamma_{1})m\lm(g_{2}),\lm(g_{1})m\s(\gamma_{2})\},
\end{equation}
for $g_{1}^{[\gamma_{1}]}, g_{2}^{[\gamma_{2}]} \in G^{[\Sigma]}$.
It may happen that this set contains infinitely many module monomials which do not divide each other, and indeed this will be the case for all sufficiently non-trivial ideals.
It implies that for such ideals, $S$ does not admit a finite Gröbner basis.

However, Lemma~\ref{lemma syzygy module} shows that a Gröbner basis of $S$ is given by adding to $\Htriv$ all the syzygies found by regular $\s$-reducing to zero all regular S-polynomials of $G^{[\Sigma]}$.
Furthermore, if $G^{[\Sigma]}$ is finite, the set of signatures of syzygies in $\Htriv$ can be enumerated using the description \eqref{eq:1}, and altogether, we obtain an effective description of the syzygy module of $f_{1},\dots,f_{r}$.

\subsection{Algorithm}
\label{sec:algorithm-1}

The algorithm, incorporating both the computation of the labelled Gröbner basis and of the aforementioned description of the syzygy module of $f_{1},\dots,f_{r}$, is given in Algorithm~\ref{algo LabelledGB}.
\chnote{added a note here}
We note that we state this algorithm only for theoretical consideration.
In an actual implementation, one would replace all computations with labelled polynomials in Algorithm~\ref{algo LabelledGB} by computations with signature polynomials.
In Section~\ref{sec:reconstruction}, we state with Algorithm~\ref{algo SigGB} an optimized version of Algorithm~\ref{algo LabelledGB} incorporating this.


\begin{algorithm}[h]
\renewcommand{\algorithmicrequire}{\textbf{Input:}}
\renewcommand{\algorithmicensure}{\textbf{Output (if the algorithm terminates):}}
\caption{LabelledGB}\label{algo LabelledGB}
\begin{algorithmic}[1]
\Require{$(f_1,\dots,f_r) \in K\<X>^{r}$ generating an ideal $I$}
\Ensure{${}$\newline 
  \vspace{-0.45cm}
    \begin{itemize}
      \setlength{\itemindent}{-15pt}
      \item $G^{[\Sigma]}$ a labelled Gr\"obner basis of $I^{[\Sigma]}$
      \item $H \subseteq \Fr$ s.t.\ 
      $H \cup  \{ \gamma m g' - g m \gamma' \mid g^{[\gamma]}, g'^{[\gamma']} \in G^{[\Sigma]}, \,m \in \<X>\}$ is a Gr\"obner basis of $\Syz(f_1,\dots,f_r)$
    \end{itemize}
}
\State{$G^{[\Sigma]} \leftarrow \emptyset$}
\State{$H \leftarrow \emptyset$}
\State{$P \leftarrow \{f_1^{[\varepsilon_1]},\dots,f_r^{[\varepsilon_r]}\}$}
\While{$P \neq \emptyset$}
	\State{choose $p^{[\pi]} \in P$ s.t.\ $\s(\pi) = \min \{\s(\pi') \mid p'^{[\pi']} \in P\}$}\label{line choice p 1}
	\State{$P \leftarrow P\setminus \{p^{[\pi]}\}$}
	\State{$p'^{[\pi']} \leftarrow$ result of regular $\s$-reducing $p^{[\pi]}$ by $G^{[\Sigma]}$}\label{line reduction sigGB 1}
	\If{$p' = 0$}
		\State{$H \leftarrow H \cup \{\pi'\}$}
	\ElsIf{$p'^{[\pi']}$ is not singular top $\s$-reducible by $G^{[\Sigma]}$} \label{line singular check 1}
		\State{$G^{[\Sigma]} \leftarrow G^{[\Sigma]} \cup \{p'^{[\pi']}\}$} \label{line add to G 1}
		\State{$P \leftarrow P \cup \{$all regular S-polynomials between $p'^{[\pi']}$ and all $g^{[\gamma]} \in G^{[\Sigma]}\}$ }\label{line form s-polies 1}
	\EndIf
\EndWhile
\State{\textbf{return }$G^{[\Sigma]},H$}
\end{algorithmic}
\end{algorithm}

We note that we cannot expect Algorithm~\ref{algo LabelledGB} to always terminate since, as already mentioned, there are polynomials in $K\<X>$ generating a module $I^{[\Sigma]}$ which does not have a finite labelled Gr\"obner basis.
However, the following theorem ensures that the algorithm always correctly enumerates a labelled Gr\"obner basis of the module $I^{[\Sigma]}$ defined by the input $(f_1,\dots,f_r) \in K\<X>^r$, and a Gr\"obner basis of the syzygy module $\Syz(f_1,\dots,f_r)$.

\begin{theorem}\label{thm algo LabelledGB correct}
Let $(f_1,\dots,f_r) \in K\<X>^r$, denote $G_0^{[\Sigma]} = H_0 = \emptyset$.
Furthermore, let $G_n^{[\Sigma]}$ and $H_n$ be the value of $G^{[\Sigma]}$ and $H$ in Algorithm~\ref{algo LabelledGB} after $n$ iterations of the \textwhile loop given $f_1,\dots,f_r$ as input.
Then, the following holds:
\begin{enumerate}
	\item $G^{[\Sigma]} = \bigcup_{n \geq 0} G_n^{[\Sigma]}$ is a labelled Gr\"obner basis of $I^{[\Sigma]}$ w.r.t.\ the family of generators $f_1,\dots,f_r$;
	\item Let $\Htriv =  \{ \gamma m g' - g m \gamma' \mid g^{[\gamma]}, g'^{[\gamma']} \in G^{[\Sigma]}\}$. Then,
	\[
	H \cup \Htriv =\bigcup_{n\geq 0} \big( H_n \cup \{ \gamma m g' - g m \gamma' \mid g^{[\gamma]}, g'^{[\gamma']} \in G_n^{[\Sigma]}, m \in \<X>\}\big)
	\]
	is a Gr\"obner basis of $\Syz(f_1,\dots,f_r)$.
\end{enumerate}
In this sense, Algorithm~\ref{algo LabelledGB} enumerates a labelled Gr\"obner basis of $I^{[\Sigma]}$ and a Gr\"obner basis of the syzygy module $\Syz(f_1,\dots,f_r)$.
\end{theorem}

In order to prove this theorem, we first state the following useful lemma which ensures that Algorithm~\ref{algo LabelledGB} cannot ``get stuck'' at a certain signature indefinitely.

\begin{lemma}\label{lemma ascending signatures}
During the execution of Algorithm~\ref{algo LabelledGB}, elements from $P$ are processed in ascending order w.r.t.\ their signatures and every possible signature is eventually processed.
\end{lemma}

\begin{proof}
\chnote{Simplified proof.}
Note that the set $P$ is finite at all times.
Hence, we can associate to $P$ the tuple $P'$ of all signatures of $P$ sorted in increasing order, i.e., if $P = \{p_1^{[\pi_1]},\dots,p_n^{[\pi_n]}\}$ with $\s(\pi_1) \preceq \dots \preceq \s(\pi_n)$, then $P' = (\s(\pi_1), \dots, \s(\pi_n))$.
We claim that $P'$ strictly increases lexicographically between each run of line~\ref{line choice p 1}.
If the algorithm does not reach line~\ref{line form s-polies 1} after choosing $p^{[\pi]}\in P$ in line~\ref{line choice p 1}, this statement follows since $p^{[\pi]}$ is removed but nothing is added to $P$.
Otherwise, the statement follows from Lemma~\ref{lemma signature s-poly}, which implies that the signatures of the elements added to $P$ in line~\ref{line form s-polies 1} are strictly larger than $\s(\pi)$.
Note that Lemma~\ref{lemma signature s-poly} is applicable here since the algorithm only reaches line~\ref{line form s-polies 1} if the normal form $p'^{[\pi']}$ computed in line~\ref{line reduction sigGB 1} is not top $\s$-reducible by $G^{[\Sigma]}$ (before $p'^{[\pi']}$ is added to $G^{[\Sigma]}$).
Then, $P'$ strictly increasing between each run of line~\ref{line choice p 1} and the fact that $p^{[\pi]}$ is always chosen to have minimal signature among all elements in $P$ show that elements from $P$ are processed in ascending order w.r.t.\ their signature.
Furthermore, this, together with the fairness of the module ordering, also implies that every element in $P$ will be removed eventually.

\end{proof}

Using this lemma, we can now proceed to prove Theorem~\ref{thm algo LabelledGB correct}.

\begin{proof}[Proof of Theorem~\ref{thm algo LabelledGB correct}]
We prove each claim separately.

\paragraph{$G^{[\Sigma]}$ is a labelled Gr\"obner basis of $I^{[\Sigma]}$}
For $n \geq 0$, we denote by $p_n^{[\pi_n]}$ the labelled polynomial which is chosen from the set $P$ in the $(n+1)$-th iteration of the \textwhile loop of Algorithm~\ref{algo LabelledGB}.
We claim that for every $n \geq 0$, the set $G_n^{[\Sigma]}$ is a labelled Gr\"obner basis up to signature $\s(\pi_n)$.
Indeed, it follows from Lemma~\ref{lemma ascending signatures}, that, when $p_n^{[\pi_n]}$ is chosen in line~\ref{line choice p 1}, all S-polynomials as well as all $f_i^{[\varepsilon_i]}$ with signature strictly smaller than $\s(\pi_n)$ have already been processed. 
Therefore, Theorem~\ref{thm spairs} yields that $G_n^{[\Sigma]}$ is a labelled Gr\"obner basis up to signature $\s(\pi_n)$.
Since every signature is eventually processed, the set $G^{[\Sigma]} = \bigcup_{n \geq 0} G_n^{[\Sigma]}$ is a labelled Gr\"obner basis of $I^{[\Sigma]}$.

\paragraph{$H \cup \Htriv$ is a Gr\"obner basis of $\Syz(f_1,\dots,f_r)$}
Let $\mu \in \Syz(f_1,\dots,f_r) \setminus\{0\}$.
We have to show that there exist $n \geq 0$ and $a,b \in \<X>$ such that $\s(\mu) = a \s(\alpha) b$ for some $\alpha \in H_n \cup \{\gamma m g' - g m \gamma' \mid g^{[\gamma]}, g'^{[\gamma']} \in G_n^{[\Sigma]}, m \in \<X>\}$.
To this end, let $n$ be such that $G_n^{[\Sigma]}$ is a labelled Gr\"obner basis up to a signature $\s(\mu)$.
Such an $n$ must exist due to the previous discussion. 
Then, Lemma~\ref{lemma syzygy module} yields the existence of $p^{[\pi]} \in I^{[\Sigma]}$ and $a,b \in \<X>$ such that $\s(a \pi b) = \s(\mu)$ and such that one of the following conditions holds:
\begin{enumerate}
	\item $\pi$ is a trivial syzygy between two elements in $G_n^{[\Sigma]}$;
	\item $p^{[\pi]}$ is an S-polynomial of $G_n^{[\Sigma]}$ which regular $\s$-reduces to zero by $G_n^{[\Sigma]}$;
\end{enumerate}
Hence, either a module element with signature $\s(\pi)$ has been added to $H_n$ or there exist $g^{[\gamma]}$, $g'^{[\gamma']} \in G_n^{[\Sigma]}$ and $m \in \<X>$ such that $\s(\pi) = \s(\gamma m g' - g m \gamma')$.
\end{proof}

We make a few observations about Algorithm~\ref{algo LabelledGB}.
\begin{enumerate}
	\item As already mentioned in Lemma~\ref{lemma ascending signatures}, Algorithm~\ref{algo LabelledGB} processes S-polynomials in ascending order w.r.t.\ their signature.
	Furthermore, the requirement that $\preceq$ is a fair module ordering ensures that no S-polynomial is postponed indefinitely and consequently enforces a fair selection strategy.
	Both of these properties are crucial to ensure the correctness of the algorithm.
	
        \item It is possible that at some point in the algorithm the set $\{g \mid g^{[\gamma]}\in G^{[\Sigma]}\}$ is a Gröbner basis, even if $G^{[\Sigma]}$ is not a labelled Gröbner basis, or even if $I^{[\Sigma]}$ does not admit a finite labelled Gröbner basis. For instance,
        using the family of generators $f_1,f_2,f_3$ from Example~\ref{infinite-sgb} as input to Algorithm~\ref{algo LabelledGB}, after processing the signature $\varepsilon_{1}y$, the set $G^{[\Sigma]}$ becomes $G^{[\Sigma]} = \{f_1^{[\varepsilon_{1}]},f_2^{[\varepsilon_{2}]},f_3^{[\varepsilon_{3}]},f_4^{[\alpha]}\}$, and the set of corresponding polynomials is a Gröbner basis of $I = (f_1,f_2,f_3)$, although $G^{[\Sigma]}$ is not a labelled Gröbner basis of the module $I^{[\Sigma]}$.

        We do not know whether this happens whenever an ideal $I$ admits a finite Gröbner basis.

        \item After every iteration of the \textwhile loop, the set $G^{[\Sigma]}$ is a labelled Gr\"obner basis up to signature $\sigma$, 
	where $\sigma \in \MFr$ is the minimal signature of all elements left in $P$.
	
	\item Whenever an element $p'^{[\pi']}$ is added to $G^{[\Sigma]}$ in line~\ref{line add to G 1}, it is not top $\s$-reducible by $G^{[\Sigma]}$. 
	Also, no element $q'^{[\rho']}$, which is added to $G^{[\Sigma]}$ after $p'^{[\pi']}$, can be used to top $\s$-reduce $p'^{[\pi']}$.
	To see this, we note that it follows from the first point above that $\s(\pi') \preceq \s(\rho')$.
	Now, if $\s(\pi') \prec \s(\rho')$, then $q'^{[\rho']}$ can obviously not be used to $\s$-reduce $p'^{[\pi']}$.
	If $\s(\pi') = \s(\rho')$, then $q'^{[\rho']}$ can only be used to top $\s$-reduce $p'^{[\pi']}$ if $\lm(q') = \lm(p')$. 
	But this would imply that $q'^{[\rho']}$ is singular top $\s$-reducible by $p'^{[\pi']}$, which would contradict the check in line~\ref{line singular check 1}.
\end{enumerate}

The following corollary is an immediate consequence of the last observation.
\begin{corollary}\label{cor:enum-min-sGB}
Algorithm~\ref{algo LabelledGB} enumerates a minimal labelled Gr\"obner basis.
\end{corollary}

Combining this corollary with Corollary~\ref{cor finite s-GB}, we see that Algorithm~\ref{algo LabelledGB} terminates whenever $I^{[\Sigma]}$ admits a finite labelled Gr\"obner basis w.r.t.\ the family of generators $f_1,\dots,f_r$.



\begin{corollary}\label{cor termination algo 1}
Let $(f_1,\dots,f_r) \in K\<X>^r$ be such that the corresponding module $I^{[\Sigma]}$ has a finite labelled Gröbner basis.
Then, Algorithm~\ref{algo LabelledGB} terminates when given $f_1,\dots,f_r$ as input. 
\end{corollary}
\begin{proof}
  Since $I^{[\Sigma]}$ has a finite labelled Gröbner basis, by Corollary~\ref{cor:enum-min-sGB}, the algorithm will eventually compute a minimal labelled Gröbner basis $G^{[\Sigma]}$ of $I^{[\Sigma]}$, which must be finite as well by Corollary~\ref{cor finite s-GB}.
  At each run of the loop, only finitely many S-polynomials are added to $P$, so $P$ has finite cardinality.
  Since $G^{[\Sigma]}$ is a labelled Gröbner basis, all the remaining elements in $P$ will regular $\s$-reduce to $0$ or to a singular top $\s$-reducible normal form, so no new polynomials will be added to $P$ and the algorithm will terminate.
\end{proof}

Note that if the algorithm terminates, or equivalently if $I^{[\Sigma]}$ admits a finite labelled Gröbner basis, then it has finite output $G^{[\Sigma]}$ and $H$.
This output is such that the polynomial part of elements of $G^{[\Sigma]}$ forms a (finite) Gröbner basis of the ideal $I$, and that $H \cup \Htriv$ is a (usually infinite, but with a finite data representation) Gröbner basis of the module $\Syz(f_1,\dots,f_r)$.

We conjecture that also the converse holds.
\newcommand{\Gt}{\tilde{G}}
\newcommand{\Ht}{\tilde{H}}
\begin{conjecture}
  \label{sec:conj-termination}
  Let $(f_1,\dots,f_r) \in K\<X>^r$ and $I^{[\Sigma]}$ be the corresponding module.
  Assume that there exists a finite set $\Gt^{[\Sigma]} \subseteq I^{[\Sigma]}$ of labelled polynomials, with $f_{1}^{[\varepsilon_{1}]},\dots,f_{r}^{[\varepsilon_{r}]} \in \Gt^{[\Sigma]}$, and a finite subset $\Ht \subseteq \Syz(f_1,\dots,f_r)$ such that:
    \begin{itemize}
      \item $\{g \mid g^{[\gamma]} \in \Gt^{[\Sigma]}\}$ is a Gröbner basis of $I$;
      \item $\Ht \cup \{\gamma m g' - g m \gamma' \mid g^{[\gamma]}, g'^{[\gamma']} \in \Gt^{[\Sigma]}\}$ is a Gröbner basis of $\Syz(f_{1},\dots,f_{r})$;
    \end{itemize}
    Let $G_n^{[\Sigma]} \subseteq I^{[\Sigma]}$ and $H_n \subseteq \Syz(f_1,\dots,f_r)$ be intermediate values of $G^{[\Sigma]}$ and $H$, respectively, in Algorithm~\ref{algo LabelledGB} such that
    \begin{itemize}
      \item all elements of $\Gt^{[\Sigma]}$ $\s$-reduce to $0$ modulo $G_n^{[\Sigma]}$;
      \item for all $\tilde{\sigma} \in \Ht$, there exists $\sigma \in H_n$ such that $\s(\sigma) = \s(\tilde{\sigma})$;
    \end{itemize}
    Then $G_n^{[\Sigma]}$ is a labelled Gröbner basis of $I^{[\Sigma]}$, and in particular, $I^{[\Sigma]}$ has a finite labelled Gröbner basis.
\end{conjecture}
Note that $\Gt^{[\Sigma]}$ need not be a labelled Gröbner basis, but merely a set of labelled polynomials which, without the module representations, forms a Gröbner basis.
Put differently, the statement is equivalent to saying that $I$ admits a finite Gröbner basis $G$, and that the module $\Syz(f_{1},\dots,f_{r})$ has a Gröbner basis given by adding a finite set to the set of trivial syzygies of $G$ (expressed in the module $\Fr$).

In the commutative case where all ideals have a finite signature Gröbner basis, the analogue of this conjecture would give a characterization of signature Gröbner bases in terms of a Gröbner basis of the ideal and of its module of syzygies.
To the best of our knowledge, no such characterization is proved in the commutative case.

\subsection{S-polynomial elimination}\label{sec elimination criteria}

In the commutative case, it is well known that additional criteria can be used to detect \mbox{$\s$-reductions} to zero.
So far, we have already seen that we can immediately discard all singular S-polynomials and remove a regular S-polynomial if it leads to a singular top $\s$-reducible normal form.
In this section, we adapt some other well-known techniques from the commutative case to our setting, namely the \emph{syzygy criterion}, the \emph{F5 criterion} and the \emph{singular criterion}.
In Algorithm~\ref{algo SigGB}, we include these criteria to show how to use them in practice.

\begin{proposition}[\textbf{Syzygy criterion}]\label{prop syzygy criterion}
Let $p^{[\pi]} \in I^{[\Sigma]}$ and let $G^{[\Sigma]} \subseteq I^{[\Sigma]}$ be a labelled Gr\"obner basis up to signature $\s(\pi)$.
If there exists a syzygy $\sigma \in \Fr$ and $a,b \in \<X>$ such that $\s(\pi) = a\s(\sigma)b$, then $p^{[\pi]}$ can be regular $\s$-reduced to zero by $G^{[\Sigma]}$.
\end{proposition}

\begin{proof}
Let $\sigma \in \Fr$ be a syzygy and  $a,b \in \<X>$ such that $\s(\pi) = a\s(\sigma)b$.
Now, consider
\[
	p^{[\tau]} = p^{[\pi]} - \frac{\sigc(\pi)}{\sigc(\sigma)} a 0^{[\sigma]} b.
\]
Then, $\s(\tau) \prec \s(\pi)$.
Since $G^{[\Sigma]}$ is a labelled Gr\"obner basis up to signature $\s(\pi)$, the labelled polynomial $p^{[\tau]}$ $\s$-reduces to zero by $G^{[\Sigma]}$.
Thus, using the same reductions, we see that $p^{[\pi]}$ regular $\s$-reduces to zero by $G^{[\Sigma]}$.
\end{proof}

Hence, we can immediately discard an S-polynomial $p^{[\pi]}$ during the computation of a labelled Gr\"obner basis if its signature $\s(\pi)$ is divisible by the signature of a syzygy.
Clearly, we obtain syzygies whenever we $\s$-reduce an S-polynomial to zero but there are also syzygies known prior to any computations. 
Recall that for all labelled polynomials $f^{[\alpha]}, g^{[\beta]} \in I^{[\Sigma]}$ we have the trivial syzygies $\alpha m g -  f m \beta$, for all monomials $m \in \<X>$.
This means that for any family of generators $f_1,\dots,f_r \in K\<X>$, we immediately obtain the trivial syzygies
\[
	\varepsilon_i m f_j - f_i m \varepsilon_j,
\]
for all $1 \leq i, j \leq r$ and all $m \in \<X>$, which we can use to eliminate S-polynomials.
Additionally, whenever we add a new element $g^{[\gamma]}$ to $G^{[\Sigma]}$ during the executing of Algorithm~\ref{algo LabelledGB},
we get the new trivial syzygies
\[
	\gamma m g' - g m \gamma' \quad \text{ and } \quad \gamma' m g - g' m \gamma
\]
for all $g'^{[\gamma']} \in G^{[\Sigma]}$ and all $m \in \<X>$.
Identifying those trivial syzygies leads to the \emph{F5 criterion}.

\begin{corollary}[\textbf{F5 criterion}]
  \label{prop:f5-criterion}
  Let $p^{[\pi]} \in I^{[\Sigma]}$ and let $G^{[\Sigma]} \subseteq I^{[\Sigma]}$ be a labelled Gr\"obner basis up to signature $\s(\pi)$.
  Assume that there exist $g^{[\gamma]}, g'^{[\gamma']} \in G^{[\Sigma]}$ and $a,b,m \in \<X>$ such that 
  one of the following conditions holds:
  \begin{enumerate}
    \item $\s(\pi) = a\s(\gamma) m \lm(g')b\;\;$ and $\;\;\s(\gamma) m \lm(g') \succ \lm(g) m \s(\gamma')$;
    \item $\s(\pi) = a\lm(g) m \s(\gamma')b\;\;$ and $\;\;\lm(g) m \s(\gamma') \succ \s(\gamma) m \lm(g')$;
  \end{enumerate}
  Then $p^{[\pi]}$ can be regular $\s$-reduced to zero by $G^{[\Sigma]}$.
\end{corollary}

\begin{remark}
  It is not clear whether it is possible to check the F5 criterion in the noncommutative case as efficiently as in the commutative case.
  \emph{A priori}, it requires $\Theta(|G^{[\Sigma]}|^2)$ checks of the conditions of Corollary~\ref{prop:f5-criterion}.
\end{remark}

Lemma~\ref{lemma same signature} provides another way to detect redundant S-polynomials.

\begin{corollary}[\textbf{Singular criterion}]\label{lemma singular criterion}
Let $p^{[\pi]} \in I^{[\Sigma]}$ and let $G^{[\Sigma]} \subseteq I^{[\Sigma]}$ be a labelled Gr\"obner basis up to signature $\s(\pi)$.
If there exists a regular $\s$-reduced element $g^{[\gamma]} \in G^{[\Sigma]}$ such that $\s(\gamma) = \s(\pi)$, then $p^{[\pi]}$ regular $\s$-reduces to a normal form that is singular top $\s$-reducible by $G^{[\Sigma]}$.
\end{corollary}

\begin{proof}
Assume that there exists a regular $\s$-reduced element $g^{[\gamma]} \in G^{[\Sigma]}$ such that $\s(\gamma) = \s(\pi)$.
It follows from Lemma~\ref{lemma same signature} that $p^{[\pi]}$ regular $\s$-reduces to $cg^{[\pi']}$ for some scalar $c \in K$.
Since $\s(\pi') = \s(\pi) = \s(\gamma)$, this normal form is singular top $\s$-reducible by $g^{[\gamma]} \in G^{[\Sigma]}$.
\end{proof}

Using Algorithm~\ref{algo LabelledGB}, all elements that are added to $G^{[\Sigma]}$ are regular $\s$-reduced. 
Hence, during the execution of Algorithm~\ref{algo LabelledGB}, an S-polynomial $p^{[\pi]}$ can be removed immediately if its signature already appears in $G^{[\Sigma]}$.

\section{Computation of signature Gröbner bases and reconstruction}
\label{sec:reconstruction}

So far, Algorithm~\ref{algo LabelledGB} keeps but does not exploit all the information encoded in the full module representation of the polynomials. 
As indicated earlier, keeping track of the full module representation, however, causes a significant overhead in terms of memory consumption and overall computation time.
Consequently, in an actual implementation of Algorithm~\ref{algo LabelledGB}, one would only keep track of the signatures of each polynomial, and thereby, work with signature polynomials.
In doing so, instead of computing a labelled Gr\"obner basis and a Gr\"obner basis of $\Syz(f_1,\dots,f_r)$, the algorithm only computes a signature Gr\"obner  basis (Definition~\ref{def sig-labelled GB}) and a Gr\"obner basis of the module generated by $\s(\Syz(f_1,\dots,f_r))$.
Additionally, to obtain an efficient implementation, one would also exploit the elimination criteria discussed in the previous section.
Note that these criteria only depend on information encoded in signature polynomials.
Incorporating these changes leads to Algorithm~\ref{algo SigGB}, which is an optimized version of Algorithm~\ref{algo LabelledGB}.

\begin{algorithm}[h]
\renewcommand{\algorithmicrequire}{\textbf{Input:}}
\renewcommand{\algorithmicensure}{\textbf{Output (if the algorithm terminates):}}
\caption{SigGB}\label{algo SigGB}
\begin{algorithmic}[1]
\Require{$(f_1,\dots,f_r) \in K\<X>^{r}$ generating an ideal $I$}
\Ensure{${}$\newline 
  \vspace{-0.45cm}
    \begin{itemize}
      \setlength{\itemindent}{-15pt}
      \item $G^{(\Sigma)}$ a minimal signature Gr\"obner basis of $I^{[\Sigma]}$
      \item $H \subseteq \MFr$ s.t.\ 
      $H \cup  \{ \max \{\sigma m \lm(g'), \lm(g) m \sigma'\} \mid g^{(\sigma)}, g'^{(\sigma')} \in G^{(\Sigma)}, \,m \in \<X>\}$ is a Gr\"obner basis of the module generated by $\s(\Syz(f_1,\dots,f_r))$
    \end{itemize}
}

\State{$G^{(\Sigma)} \leftarrow \emptyset$}
\State{$H \leftarrow \emptyset$}
\State{$P \leftarrow \{f_1^{(\varepsilon_1)},\dots,f_r^{(\varepsilon_r)}\}$}

\While{$P \neq \emptyset$}\label{line start while}
	\State{choose $p^{(\sigma)} \in P$ s.t.\ $\sigma = \min \{\sigma' \mid p'^{(\sigma')} \in P\}$}\label{line choice p}
	\State{$P \leftarrow P\setminus \{p^{(\sigma)}\}$}
	\If{$p^{(\sigma)}$ satisfies the hypotheses of the Syzygy criterion (Prop.~\ref{prop syzygy criterion}),\\
          \hspace{1.6em}\hphantom{\textbf{if}} the F5 criterion (Cor.~\ref{prop:f5-criterion}), or the Singular criterion (Cor.~\ref{lemma singular criterion})}
        \State{\textbf{goto} line~\ref{line start while}}
        \Else
	\State{$p'^{(\sigma)} \leftarrow$ result of regular $\s$-reducing $p^{(\sigma)}$ by $G^{(\Sigma)}$}\label{line reduction sigGB}
	\If{$p' = 0$}
		\State{$H \leftarrow H \cup \{\sigma\}$}
	\ElsIf{$p'^{(\sigma)}$ is not singular top $\s$-reducible by $G^{(\Sigma)}$} \label{line singular check}
		\State{$G^{(\Sigma)} \leftarrow G^{(\Sigma)} \cup \{p'^{(\sigma)}\}$} \label{line add to G}
		\State{$P \leftarrow P \cup \{$all regular S-polynomials between $p'^{(\sigma)}$ and all $g^{(\gamma)} \in G^{(\Sigma)}\}$ }\label{line form s-polies}
	\EndIf
	\EndIf
\EndWhile
\State{\textbf{return }$G^{(\Sigma)},H$}
\end{algorithmic}
\end{algorithm}

\begin{theorem}
Algorithm~\ref{algo SigGB} is correct.
Furthermore, if $(f_1,\dots,f_r) \in K\<X>^r$ is such that the corresponding module $I^{[\Sigma]}$ has a finite labelled Gröbner basis,
then Algorithm~\ref{algo SigGB} terminates when given $f_1,\dots,f_r$ as input. 
\end{theorem}

\begin{proof}
Follows from the correctness of Algorithm~\ref{algo LabelledGB}, Corollary~\ref{cor termination algo 1} and Section~\ref{sec elimination criteria}.
\end{proof}

In the following, we discuss how to recover the information that is lost when Algorithm~\ref{algo SigGB} is used instead of Algorithm~\ref{algo LabelledGB}.
In particular, this means reconstructing a labelled Gr\"obner basis from a signature Gr\"obner basis and reconstructing a Gr\"obner basis of  $\Syz(f_1,\dots,f_r)$ from one of the module generated by $\s(\Syz(f_1,\dots,f_r))$. 
To this end, we adapt the reconstruction methods described in~\cite{Gao-2015-new-framework-for} to recover module representations of elements of the ideal and of syzygies from signatures to our noncommutative setting.

We let $G^{(\Sigma)} \subseteq I^{(\Sigma)}$ and $H \subseteq \MFr$ be the output of Algorithm~\ref{algo SigGB} when given the family of generators $f_1,\dots,f_r \in K\<X>$ as input.
Recall that the algorithm does not necessarily terminate.
As such, $G^{(\Sigma)}$ will either be the full output of Algorithm~\ref{algo SigGB} assuming termination, or the partial output after interrupting the computation.
In the latter case, the set $G^{(\Sigma)}$ is only a signature Gr\"obner basis up to a certain signature $\sigma \in \MFr$ and $H$ together with the signatures of the trivial syzygies does not necessarily form a Gr\"obner basis of the module generated by $\s(\Syz(f_1,\dots,f_r))$.

In this general setting, the goal of this section is twofold.
First of all, starting from $G^{(\Sigma)}$ we want to reconstruct a labelled Gr\"obner basis $G^{[\Sigma]}$ (up to signature $\sigma$).
Secondly, for each element $\beta \in H$, we want to find a module element $\alpha \in \Syz(f_1,\dots,f_r)$ such that $\s(\alpha) = \beta$.

In situations where Algorithm~\ref{algo SigGB} terminates, i.e., when $G^{(\Sigma)}$ is a signature Gr\"obner basis and $H$ together with the signatures of the trivial syzygies forms a Gr\"obner basis of the module generated by $\s(\Syz(f_1,\dots,f_r))$, achieving both of these goals allows us to also recover a Gr\"obner basis of $\Syz(f_1,\dots,f_r)$.
The algorithms which we describe in this section are a direct adaptation of the procedure outlined in~\cite{Gao-2015-new-framework-for}.

Our first goal can be achieved by the following algorithm.
We note that no matter whether Algorithm~\ref{algo SigGB} terminates by itself or whether we interrupt the computation, the sets $G^{(\Sigma)}$ and $H$ are always finite.

\begin{algorithm}[H]
\renewcommand{\algorithmicrequire}{\textbf{Input:}}
\renewcommand{\algorithmicensure}{\textbf{Output:}}
\caption{Sig2LabelledGB}\label{algo 2}
\begin{algorithmic}[1]
\Require{$G^{(\Sigma)}$ a finite minimal signature Gr\"obner basis (up to some signature $\mu \in \MFr$)}
\Ensure{$G^{[\Sigma]}$ a finite minimal labelled Gr\"obner basis (up to signature $\mu$)}

\State{$G^{[\Sigma]} \leftarrow \emptyset$}
\State{$H^{(\Sigma)} \leftarrow G^{(\Sigma)}$}
\Comment{make a copy so that we do not alter $G^{(\Sigma)}$}

\While{$H^{(\Sigma)} \neq \emptyset$}\label{line while rec}
	\State{choose $f^{(\sigma)} \in H^{(\Sigma)}$ s.t.\ $\sigma = \min \{\sigma' \mid f'^{(\sigma')} \in H^{(\Sigma)}\}$}\label{line f}
	\State{$H^{(\Sigma)} \leftarrow H^{(\Sigma)}\setminus \{f^{(\sigma)}\}$}
	\State{choose $a,b\in\<X>, g^{[\gamma]} \in G^{[\Sigma]} \cup \{f_1^{[\varepsilon_1]}, \dots, f_r^{[\varepsilon_r]}\}$ s.t.\ $\s(a\gamma b) = \sigma$ and $\lm(agb)$ is minimal}\label{line choose 1}
	\State{$g'^{[\gamma']} \leftarrow$ result of regular top $\s$-reducing $ag^{[\gamma]}b$ by $G^{[\Sigma]}$}\label{line reduce 1}
	\State{$G^{[\Sigma]} \leftarrow G^{[\Sigma]} \cup \{g'^{[\gamma']}\}$}
\EndWhile
\State{\textbf{return }$G^{[\Sigma]}$}
\end{algorithmic}
\end{algorithm}

\begin{remark}
As will be clear from the proof of Proposition~\ref{prop correctness algo 2}, the minimality condition in line~\ref{line choose 1} of Algorithm~\ref{algo 2} is not required for the correctness of the algorithm. 
It is included purely for efficiency reasons with the hope of having to do less $\s$-reductions if $\lm(agb)$ is minimal.
We note that the same also holds for the minimality condition in line~\ref{line choose 2} of Algorithm~\ref{algo 3}.
\end{remark}

\begin{proposition}\label{prop correctness algo 2}
Algorithm~\ref{algo 2} is correct.
\end{proposition}

\begin{proof}
Let $\tilde G^{[\Sigma]} \subseteq I^{[\Sigma]}$ be a minimal labelled Gr\"obner basis of $I^{[\Sigma]}$ (up to signature $\mu$) such that $G^{(\Sigma)} = \{g^{(\s(\gamma))} \mid g^{[\gamma]} \in \tilde G^{[\Sigma]}\}$.
Furthermore, let $G^{[\Sigma]}$ be the output of Algorithm~\ref{algo 2} given $G^{(\Sigma)}$ as input.
To prove the correctness of Algorithm~\ref{algo 2}, we show that
\begin{align}\label{eq proof algo 2}
	\{(\lm(g), \s(\gamma)) \mid g^{[\gamma]} \in G^{[\Sigma]} \} = \{(\lm(g), \s(\gamma)) \mid g^{[\gamma]} \in \tilde G^{[\Sigma]} \}.
\end{align}
In other words, we show that the labelled polynomials in $G^{[\Sigma]}$ have the same leading monomials and signatures as the elements in the labelled Gr\"obner basis $\tilde G^{[\Sigma]}$.
Because then, every $f^{[\alpha]} \in I^{[\Sigma]}$ is $\s$-reducible by $G^{[\Sigma]}$ if and only if it is $\s$-reducible by $\tilde G^{[\Sigma]}$ and
Lemma~\ref{lemma equiv characterisation s-GB} yields that $G^{[\Sigma]}$ is a labelled Gr\"obner basis (up to signature $\mu$).
Furthermore, the minimality of $\tilde G^{[\Sigma]}$ implies the minimality of $G^{[\Sigma]}$.

To prove~\eqref{eq proof algo 2}, we show that the following loop invariant holds whenever the algorithm reaches line~\ref{line while rec}:
\begin{equation}
	\label{loop invariant}
	\{(\lm(g), \sigma) \mid g^{(\sigma)} \in G^{(\Sigma)} \setminus H^{(\Sigma)} \} = \{(\lm(g), \s(\gamma)) \mid g^{[\gamma]} \in G^{[\Sigma]}\}.
\end{equation}
Once the algorithm terminates and $H^{(\Sigma)} = \emptyset$ this implies~\eqref{eq proof algo 2} since the leading monomials and signatures of the elements in $G^{(\Sigma)}$ are equal to those of  $\tilde G^{[\Sigma]}$ by definition of $\tilde G^{[\Sigma]}$.

Obviously~\eqref{loop invariant} holds in the very beginning when $H^{(\Sigma)} = G^{(\Sigma)}$.
So, now assume that~\eqref{loop invariant} holds at some point when the algorithm reaches line~\ref{line while rec} and let $f^{(\sigma)} \in H^{(\Sigma)}$ be the signature polynomial that is chosen in line~\ref{line f}.
Furthermore, let $\alpha \in \Fr$ be such that $f^{[\alpha]} \in \tilde G^{[\Sigma]}$ with $\s(\alpha) = \sigma$.
Then, let $a,b \in \<X>$ and $g^{[\gamma]} \in G^{[\Sigma]} \cup \{f_1^{[\varepsilon_1]}, \dots, f_r^{[\varepsilon_r]}\}$ be as chosen in line~\ref{line choose 1}.
Due to the presence of the generators $f_1^{[\varepsilon_1]}, \dots, f_r^{[\varepsilon_r]}$ such a choice of $a,b$ and $g^{[\gamma]}$ is always possible.
Let $g'^{[\gamma']}$ be the result of the computation in line~\ref{line reduce 1}.
By construction, $g'^{[\gamma']}$ is regular top $\s$\nobreakdash-reduced by $G^{[\Sigma]}$.
Furthermore, note that $f^{[\alpha]}$ is regular top $\s$-reduced by $\tilde G^{[\Sigma]}$ because $\tilde G^{[\Sigma]}$ is minimal.
Consequently, the loop invariant implies that $f^{[\alpha]}$ is also regular top $\s$-reduced by $G^{[\Sigma]}$.
Note that, since we only care about \emph{regular} top $\s$-reducibility, it is irrelevant whether we consider $G^{[\Sigma]}$ before or after adding $g'^{[\gamma']}$ as $\s(\alpha) = \s(\gamma')$.
Also, note that the loop invariant, together with the fact that $\sigma$ was chosen to be minimal among all signatures in $H^{(\Sigma)}$, implies that $G^{[\Sigma]}$ is a labelled Gr\"obner basis up to signature $\sigma$.
Hence, Lemma~\ref{lemma same signature} is applicable to $g'^{[\gamma']}$ and $cf^{[\alpha]}$ with $c = \frac{\sigc(\gamma')}{\sigc(\alpha)}$.
It yields that $\lt(g') = \lt(cf)$, and consequently, $\lm(g') = \lm(f)$. 
Since also $\s(\gamma') = \s(a\gamma b) = \sigma$, the loop invariant still holds after removing $f^{(\sigma)}$ from $H^{(\Sigma)}$ and adding $g'^{[\gamma']}$ to $G^{[\Sigma]}$.
\end{proof}

After recovering a labelled Gr\"obner basis, we can proceed with the following algorithm to 
also recover the syzygies whose signatures are saved in $H$.

\begin{algorithm}[h]
\renewcommand{\algorithmicrequire}{\textbf{Input:}}
\renewcommand{\algorithmicensure}{\textbf{Output:}}
\caption{SyzygyRecovery}\label{algo 3}
\begin{algorithmic}[1]
\Require{$G^{(\Sigma)} \subseteq I^{(\Sigma)}$ and $H \subseteq \MFr$ as produced by Algorithm~\ref{algo SigGB} }
\Ensure{$\tilde H \subseteq \Syz(f_1,\dots,f_r)$ such that $\s(\tilde H) = H$}

\State{$\tilde H \leftarrow \emptyset$}
\State{$G^{[\Sigma]} \leftarrow$ apply Algorithm~\ref{algo 2} to $G^{(\Sigma)}$}

\For{$\sigma \in H$}\label{line sigma}
	\State{choose $a,b\in\<X>, g^{[\gamma]} \in G^{[\Sigma]} \cup \{f_1^{[\varepsilon_1]}, \dots, f_r^{[\varepsilon_r]}\}$ s.t.\ $\s(a\gamma b) = \sigma$ and $\lm(agb)$ is minimal}\label{line choose 2}
	\State{$0^{[\gamma']} \leftarrow$ result of regular $\s$-reducing $ag^{[\gamma]}b$ by $G^{[\Sigma]}$}\label{line reduce 2}
	\State{$\tilde H \leftarrow \tilde H \cup \{\gamma'\}$}
\EndFor
\State{\textbf{return }$\tilde H$}
\end{algorithmic}
\end{algorithm}

\begin{proposition}
Algorithm~\ref{algo 3} is correct.
\end{proposition}

\begin{proof}
To see the correctness of the algorithm, the only problematic lines are line~\ref{line choose 2} and~\ref{line reduce 2}.
To this end, let $\sigma \in H$ be the module monomial that is chosen in line~\ref{line sigma} during some iteration.
Due to the presence of the generators $f_1^{[\varepsilon_1]}, \dots, f_r^{[\varepsilon_r]}$ in line~\ref{line choose 2} of the algorithm, a choice of $a,b$ and $g^{[\gamma]}$ as required in this line is always possible.
It remains to show that $ag^{[\gamma]}b$ really regular $\s$-reduces to zero by $G^{[\Sigma]}$.
To this end, we note that it follows from the definition of $G^{(\Sigma)}$ and Proposition~\ref{prop correctness algo 2}, that $G^{[\Sigma]}$ is a labelled Gr\"obner basis up to signature $\sigma' = \max H$.
Furthermore, by definition of $H$, we know that $\sigma$ is the signature of a syzygy. 
Consequently, we can apply Proposition~\ref{prop syzygy criterion} to conclude that $ag^{[\gamma]}b$ indeed regular $\s$-reduces to zero.
\end{proof}

To conclude this section, we note that if Algorithm~\ref{algo SigGB} terminates without interruption, a Gr\"obner basis of $\Syz(f_1,\dots,f_r)$ can be obtained as follows:
First, apply Algorithm~\ref{algo 2} to obtain a labelled Gr\"obner basis of $I^{[\Sigma]}$. Next, use Algorithm~\ref{algo 3} to get the set $\tilde H$ containing the recovered syzygies.
Finally, a Gr\"obner basis of $\Syz(f_1,\dots,f_r)$ is given by $\tilde H \cup \{ \gamma m g' - g m \gamma' \mid g^{[\gamma]}, g'^{[\gamma']} \in G^{[\Sigma]}, m \in \<X>\}$.

\section{Experimental results and future work}

In this section, we compare Algorithm~\ref{algo SigGB} to the classical Buchberger algorithm.
Since our focus is on the feasibility of signature-compatible computations and not on their efficiency, we give data about the number of S-polynomials computed and reduced as well as about the number of reductions to zero when computing (signature) Gr\"obner bases for certain benchmark examples.
The following are taken from~\cite{LL09}.
\begin{table}[H]
\centering
\begin{tabular}{ |c|c| } 
 \hline
Example & Generators of the ideal\\
\hline
\hline
\texttt{braid3} & $yxy-zyz, xyx-zxy, zxz - yzx, x^3+y^3+z^3+xyz $\\
 \hline
 \texttt{lp1} & $z^4 + yxyx - xy^2x - 3zyxz, x^3yxy - xyx, zyx - xyz + zxz$ \\
 \hline
 \texttt{lv2} & $xy+yz, x^2+xy - yx- y^2$ \\
 \hline
\end{tabular}
\end{table}

As done in~\cite{LL09}, we only compute truncated (signature) Gr\"obner bases of these homogeneous ideals.
The designated degree bounds are indicated by the number after the ``\texttt{-}'' in the name of each example in Table~\ref{table comparison}.
 So, for example \texttt{lp1-11} means that we compute a partial Gr\"obner basis of the example \texttt{lp1} up to degree 11. 
Additionally, we also consider two non-homogeneous ideals derived from finite generalized triangular groups taken from~\cite[Theorem 2.12]{RS02} as done in~\cite{Xiu12}.
Both of these ideals have finite (signature) Gr\"obner bases.
\begin{table}[H]
\centering
\begin{tabular}{ |c|c| } 
 \hline
Example & Generators of the ideal\\
\hline
\hline
\texttt{tri1} & $ x^3-1, y^2-1, (yxyxyx^2yx^2)^2 -1$\\
 \hline
\texttt{tri3} & $x^3-1, y^3 - 1, (yxyx^2)^2 -1$ \\
 \hline
\end{tabular}
\end{table}
For all examples, we fix $\preceq_{\textup{deglex}}$ as a monomial ordering where we order the indeterminates as $x \prec_{\textup{lex}} y \prec_{\textup{lex}} z$ and work over the coefficient field $\Q$.
As a module ordering, $\preceq_\textbf{top}$ is chosen.

Table~\ref{table comparison} compares the number of S-polynomials computed and reduced and the number of reductions to zero that occur while computing (truncated) (signature) Gr\"obner bases for the examples stated above.
Algorithm~\ref{algo SigGB}, denoted by \texttt{SigGB}, is compared to a vanilla Buchberger algorithm, denoted by \texttt{BB vanilla}, and to an optimized Buchberger algorithm including a noncommutative version of the chain criterion as described in~\cite[Sec.~4.5.1]{Hof20}, denoted by \texttt{BB optimized}.
For each example, we list in the column ``S-poly'' the total number of S\nobreakdash-polynomials that are computed and reduced during the execution of the respective algorithm.
Additionally, we list the total number of reductions to zero in the column ``red.\ to 0''.

We note that all algorithms are part of the \texttt{OperatorGB} package\footnote{Available at \url{https://clemenshofstadler.com/software/}} and that a \MMA\ notebook containing all computations can be obtained from the same website as the package.

\begin{table}[H]
\centering
\begin{tabular}{ |c|c|c|c|c|c|c| } 
 \hline
\multirow{2}{*}{Example} & \multicolumn{2}{c|}{\texttt{SigGB}} & \multicolumn{2}{c|}{\texttt{BB vanilla}} & \multicolumn{2}{c|}{\texttt{BB optimized}} \\
\cline{2-7}
& S-poly & red.\ to 0 & S-poly & red.\ to 0  & S-poly & red.\ to 0 \\
 \hline
 \hline
\texttt{braid3-10} & 1053 & 40 & 1154 & 661 & 1121 & 634  \\
 \hline
 \texttt{lp1-11} & 155 & 0 & 205 & 130 & 198 & 125  \\
 \hline
 \texttt{lv2-100} & 201 & 0 & 9702 & 4990 & 9702  & 4990  \\
 \hline
 \texttt{tri1} & 335 & 164 & 9435 & 8897 & 3480 & 3288  \\
 \hline
 \texttt{tri3} & 252 &136 & 2705 & 2573 & 1060 & 979  \\
 \hline
\end{tabular}
\caption{Number of S-polynomials and reduction to zero during the computation of (truncated) (signature) Gr\"obner bases for several benchmark examples.} 
\label{table comparison}
\end{table}

As can be seen, the signature-based algorithm considers fewer S-polynomials and needs fewer reductions to zero. 
In two of the examples, there are even no zero reductions at all.
However, in terms of absolute computation time, \texttt{SigGB} cannot compete with the two other algorithms.
In comparison, \texttt{SigGB} performs worst on the \texttt{tri1} benchmark example where it is about four times slower than \texttt{BB vanilla} and about ten times slower than \texttt{BB optimized} (62 sec vs.\ 16 sec vs.\ 6 sec). 
For other examples, such as \texttt{lv2-100}, the timings are closer together but still in favor of the classical Buchberger algorithm (60 sec vs.\ 43 sec vs.\ 46 sec).


This is mainly because of two reasons.
First of all, when using the F5 criterion, the number of checks that have to be done for each S-polynomial increases quadratically with the size of the set $G^{(\Sigma)}$, which becomes computationally quite intense as $G^{(\Sigma)}$ grows.
Additionally, the fact that we are restricted to regular $\s$\nobreakdash-reductions in Algorithm~\ref{algo SigGB} requires an additionally check before each $\s$-reduction. 
This cost also adds up for longer computations.

We will investigate whether it is possible to improve the performance of Algorithm~\ref{algo SigGB} to obtain a competitive algorithm in practice. 
One step towards achieving this goal could be finding ways to also allow non-fair module orderings such as a position-over-term ordering.
Additionally, future research will be focused on adapting the concepts developed in this paper to the noncommutative F4 algorithm.

We also plan to leverage the algorithms developed here to find short representations of ideal elements.
This is particularly useful when proving operator identities, where such short representations correspond to short proofs of the statement about operators.
In particular, the effective description of the syzygy module provided by a signature Gr\"obner basis might allow to compute the \emph{shortest} proof of certain operator identities.


\section*{Acknowledgements}

The first author was supported by the Austrian Science Fund (FWF) grant P32301.
The second author was supported by the Austrian Science Fund (FWF) grants P31571-N32 and P34872.
The authors thank the anonymous referees for their helpful suggestions which improved this work a lot, as well as 
Clemens G.\ Raab and Georg Regensburger for their careful reading and their valuable comments.

\bibliographystyle{alpha}
\bibliography{SigGB} 

\appendix
\section{Detailed example}\label{appendix}

\begin{example}[continues=infinite-calculations]
Recall that in Example~\ref{ex infinite s-GB} we considered the ideal $I = (f_1,f_2,f_3) \subseteq K\<X>$ with
\[
	f_1 = xyx - xy, \qquad f_2 = yxy, \qquad f_3 = xyy - xxy,
\]
over a field $K$ in the variables $X = \{x,y\}$.
Furthermore, we used $\preceq_{\textup{deglex}}$ where we ordered the indeterminates as $x \prec_{\textup{lex}} y$ and used $\preceq_{\textbf{top}}$ as a module ordering.
We claimed that the set
\[
	G^{[\Sigma]} = \{f_1^{[\varepsilon_1]}, f_2^{[\varepsilon_2]}, f_3^{[\varepsilon_3]}, f_4^{[\alpha]} \} \cup \{g_n^{[\gamma_n]} \mid n \geq 0\}
\]
with $f_4 = xxy$, $g_n = yx^{n+2}y$ and certain $\alpha, \gamma_n \in \Fr$ such that $\s(\alpha) = \varepsilon_1 y$ and $\s(\gamma_n) = y\varepsilon_3y^n$,
is a minimal labelled Gr\"obner basis of $I^{[\Sigma]}$ w.r.t.\ the family of generators $f_1,f_2,f_3$.
We postponed the verification that $G^{[\Sigma]}$ is indeed a labelled Gr\"obner basis.
We finish this proof here using Theorem~\ref{thm spairs}.
To this end, we compute all regular ambiguities of $G^{[\Sigma]}$ and regular $\s$-reduce the respective S\nobreakdash-polynomials.
We have the following regular ambiguities $a_{ij}$ between $f_i$ and $f_j$:
\begin{align*}
	a_{11} &= (xyxyx, xy, yx, f_1^{[\varepsilon_1]}, f_1^{[\varepsilon_1]}), & a_{12} &= (xyxy, x, y, f_1^{[\varepsilon_1]}, f_2^{[\varepsilon_2]}), \\
	a_{13} &= (xyxyy, xy, yy, f_1^{[\varepsilon_1]}, f_3^{[\varepsilon_3]}), & a_{14} &= (xyxxy, xy, xy, f_1^{[\varepsilon_1]}, f_4^{[\alpha]}), \\
	a_{21} &= (yxyx, y, x, f_2^{[\varepsilon_2]}, f_1^{[\varepsilon_1]}), & a_{22} &= (yxyxy,yx,xy,f_2^{[\varepsilon_2]}, f_2^{[\varepsilon_2]}),\\
	a_{23} &= (yxyy, y, y, f_2^{[\varepsilon_2]}, f_3^{[\varepsilon_3]}), & a_{32} &= (xyyxy, xy, xy, f_3^{[\varepsilon_3]}, f_2^{[\varepsilon_2]}),\\
	 a_{41} &= (xxyx, x, x, f_4^{[\alpha]}, f_1^{[\varepsilon_1]}), & a_{42} &= (xxyxy, xx,xy, f_4^{[\alpha]}, f_2^{[\varepsilon_2]}), \\
	a_{43} &= (xxyy, x, y, f_4^{[\alpha]}, f_3^{[\varepsilon_3]}).
\end{align*}
The corresponding S-polynomials are
\begin{align*}
	\spol(a_{11}) &= (-xyyx + xyxy)^{[\varepsilon_1yx - xy\varepsilon_1]}, & \spol(a_{12}) &= -xyy^{[\varepsilon_1y - x\varepsilon_2]}, \\
	\spol(a_{13}) &= (-xyyy + xyxxy)^{[\varepsilon_1yy - xy \varepsilon_3]}, & \spol(a_{14}) &= -xyxy^{[\varepsilon_1xy - xy \alpha]}, \\
	\spol(a_{21}) &= yxy^{[\varepsilon_2x - y \varepsilon_1]}, & \spol(a_{22}) &= 0^{[\varepsilon_2 xy - yx \varepsilon_2]}, \\
	\spol(a_{23}) &= yxxy^{[\varepsilon_2y - y \varepsilon_3]}, & \spol(a_{32}) &= - xxyxy^{[\varepsilon_3xy - xy \varepsilon_2]},\\
	\spol(a_{41}) &= xxy^{[\alpha x - x\varepsilon_1]}, & \spol(a_{42}) &= 0^{[\alpha xy - xx \varepsilon_2]}, \\
	\spol(a_{43}) &= xxxy^{[\alpha y - x \varepsilon_3]}.
\end{align*}

Of these S-polynomials, all but two regular $\s$-reduce to 0.
The two exceptions are $\spol(a_{12})$, which regular $\s$-reduces to
 $-xxy^{[\alpha_{12}]}$ with $\alpha_{12} = \varepsilon_1 y - x\varepsilon_2 + \varepsilon_3$,
 and $\spol(a_{23})$, which is already regular $\s$-reduced.
Note that $\s(\alpha_{12}) = \varepsilon_1y = \s(\alpha)$. Therefore, $-xxy^{[\alpha_{12}]}$ is singular top $\s$\nobreakdash-reducible by $f_4^{[\alpha]}$.
Similarly, $\spol(a_{23})$ is singular top $\s$-reducible by $g_0^{[\gamma_0]}$ since $\s(\varepsilon_2 y - y\varepsilon_3) = y\varepsilon_3 = \s(\gamma_0)$.

Additionally, for each $n \geq 0$, we have eight regular ambiguities between the $f_i$ and $g_n$.
The S\nobreakdash-polynomials of these ambiguities all regular $\s$-reduce to 0 except for
$\spol(\tilde a_{n3}) = yx^{n+3}y^{[\gamma_n y - yx^{n+1} \varepsilon_3]}$, which comes from the overlap ambiguity
\[
	\tilde a_{n3} = (yx^{n+2}yy, yx^{n+1}, y, g_n^{[\gamma_n]}, f_3^{[\varepsilon_3]})
\]
between $f_3^{[\varepsilon_3]}$ and $g_n^{[\gamma_n]}$ and which is already regular $\s$-reduced.
Note that $\spol(\tilde a_{n3})$ is singular top $\s$-reducible by $g_{n+1}^{[\gamma_{n+1}]}$ since $\s(\gamma_n y - yx^{n+1} \varepsilon_3) = \s(\gamma_n y)=  y \varepsilon_3 y^{n+1} = \s(\gamma_{n+1})$.

Finally, we also have the following regular ambiguity between $g_i$ and $g_j$ for all $i,j \geq 0$:
\[
	a'_{ij} = (yx^{i+2}yx^{j+2}y, yx^{i+2}, x^{j+2}y, g_i^{[\gamma_i]},  g_j^{[\gamma_j]}).
\]
The respective S-polynomial is $\spol(a'_{ij}) = 0^{[\gamma_i x^{j+2}y - yx^{i+2}\gamma_j]}$.

So, all regular S-polynomials of $G^{[\Sigma]}$ regular \mbox{$\s$-reduce} to $0$ or to a singular top $\s$-reducible element.
Hence, Theorem~\ref{thm spairs} yields that $G^{[\Sigma]}$ is a labelled Gr\"obner basis of $I^{[\Sigma]}$ w.r.t.\ the family of generators $f_1,f_2,f_3$.

Furthermore, we also claimed that the set 
\[
	G^{[\Sigma']} = \{f_1^{[\varepsilon_1]}, f_2^{[\varepsilon_2]}, f_3^{[\varepsilon_3]}, f_4^{[\varepsilon_4]} \} \cup \{yxxy^{[\gamma_0]} \}
\]
with $\s(\gamma_0) =  y\varepsilon_3$ is a minimal labelled Gr\"obner basis of $I^{[\Sigma']}$ w.r.t.\ the family of generators $f_1,f_2,f_3,f_4$.
To see why this is the case, we note that with $f_4$ now being a basis element, it has the signature $\varepsilon_4$ instead of $\varepsilon_1 y$.
Consequently, $f_4^{[\varepsilon_4]}$ can now be used to regular $\s$-reduce the S\nobreakdash-polynomial $\spol(\tilde a_{03}) = yxxxy^{[y\varepsilon_3 y - yx \varepsilon_3]}$ to zero.
Therefore, $G^{[\Sigma']}$ is a labelled Gr\"obner basis of $I^{[\Sigma']}$.
The minimality of $G^{[\Sigma']}$ follows from the fact that $\varepsilon_3 \prec \varepsilon_4$ as therefore neither 
$f_3^{[\varepsilon_3]}$ nor $yxxy^{[\gamma_0]}$ is $\s$-reducible by $f_4^{[\varepsilon_4]}$.
\end{example}

\end{document}